\documentclass[10pt,letterpaper]{article}
\pdfoutput=1 
\usepackage{amssymb} % AMS Fonts
\usepackage{amsmath} % AMSLaTeX
\usepackage{amsthm}
\usepackage{latexsym}
\usepackage[round]{natbib}
\usepackage[svgnames]{xcolor}
\usepackage{complexity}
\usepackage{comment}
\usepackage{xspace}
\usepackage{enumitem}
\usepackage{marginnote}
\usepackage{hyphenat}
\usepackage[ruled,vlined,linesnumbered,algosection]{algorithm2e}
\DontPrintSemicolon

\usepackage[pdfauthor={Mark Giesbrecht, Daniel S. Roche, and
  Hrushikesh Tilak}, pdftitle={Computing Sparse Multiples of
  Polynomials}, bookmarksnumbered=true,
colorlinks,linkcolor=darkblue,citecolor=darkgreen,urlcolor=darkred]{hyperref}

\definecolor{darkgreen}{rgb}{0,.35,0}
\definecolor{darkblue}{rgb}{0,0,.5}
\definecolor{darkred}{rgb}{.6,0,0}

\newcommand{\softO}{{O\mskip1mu\tilde{\,}\mskip1mu}}

\numberwithin{equation}{section}
\newtheorem{theorem}{Theorem}[section]
\newtheorem*{theorem*}{Theorem}
\newtheorem{lemma}[theorem]{Lemma}

\newtheorem{corollary}[theorem]{Corollary}
\theoremstyle{definition}
\newtheorem{definition}[theorem]{Definition}

\newtheorem{fact}[theorem]{Fact}

\newcommand{\QED}{}

\newcommand{\tinyspace}{\mspace{1mu}}

\newcommand{\norm}[1]{\left\lVert\tinyspace#1\tinyspace\right\rVert}

\newcommand{\abs}[1]{\left\lvert\tinyspace #1 \tinyspace\right\rvert}
\newcommand{\ceil}[1]{\left\lceil #1 \right\rceil}
\newcommand{\floor}[1]{\left\lfloor #1 \right\rfloor}

\newcommand{\divs}{{\mskip3mu|\mskip3mu}}

\newcommand{\ball}[2]{\mathbf{B}_n(#1,#2)}
\DeclareMathOperator{\Vol}{Vol}

\newcommand{\NONE}{\textrm{\upshape ``\texttt{NONE}''}}

\DeclareMathOperator{\diag}{\textrm{diag}}
\DeclareMathOperator{\rem}{rem}

\newcommand{\ZZ}{{\mathbb{Z}}}
\newcommand{\QQ}{{\mathbb{Q}}}
\newcommand{\RR}{{\mathbb{R}}}
\newcommand{\CC}{{\mathbb{C}}}
\newcommand{\NN}{{\mathbb{N}}}
\renewcommand{\H}{{\mathcal{H}}}
\renewcommand{\L}{{\mathcal{L}}}
\newcommand{\F}{{\mathsf{F}}}

\newcommand{\Fq}{\mathbb{F}_q}
\newcommand{\Fqd}{\mathbb{F}_{q^d}}
\newcommand{\Fqe}{\mathbb{F}_{q^e}}

\newcommand{\tSpMulFqn}{{\bf SpMul}${}_{\Fq}^{(t)}(f,n)$\xspace}
\newcommand{\OrderFqe}{{\bf Order}${}_{\Fqe}(a,n)$\xspace}
\newcommand{\BinSpMulFq}{{\bf SpMul}${}_{\Fq}^{(2)}(f)$\xspace}

\DeclareMathOperator{\lcm}{lcm}
\DeclareMathOperator{\sparsity}{sparsity}

\reversemarginpar
\makeatletter
\@twosidefalse
\makeatother

\begin{document}

\title{Computing sparse multiples of polynomials%
\thanks{The authors would like to thank the Natural Sciences and
   Engineering Research Council of Canada (NSERC), and MITACS}}

\author{Mark Giesbrecht, Daniel S. Roche, and Hrushikesh Tilak\\
Cheriton School of Computer Science, University of Waterloo}

\maketitle

\begin{abstract}
  We consider the problem of finding a sparse multiple of a
  polynomial.  Given $f\in\F[x]$ of degree $d$ over a field $\F$,
  and a desired sparsity $t$, our goal is to determine if there exists
  a multiple $h\in\F[x]$ of $f$ such that $h$ has at most $t$ non-zero
  terms,
  and if so, to find such an $h$.  When $\F=\QQ$ and $t$ is constant, we give a
  polynomial-time algorithm in $d$ and the size of coefficients in
  $h$.  When $\F$ is a finite field, we show that the problem is at
  least as hard as determining the multiplicative order of elements in
  an extension field of $\F$ (a problem thought to have complexity
  similar to that of factoring integers), and this lower bound is
  tight when $t=2$.
\end{abstract}

%-----------------------------------------------------------------------------%
\section{Introduction}
\label{sec:intro}
%-----------------------------------------------------------------------------%

Let $\F$ be a field, which will later be specified either to be the
rational numbers ($\QQ$) or a finite field with $q$ elements ($\Fq$).
We say a polynomial $h\in\F[x]$ is \emph{$t$-sparse} (or \emph{has
  sparsity $t$}) if it has at most $t$ nonzero coefficients in the
standard power basis; that is, $h$ can be written in the form
 \begin{equation}
\label{eq:hsparse}
h=h_1x^{d_1}+h_2x^{d_2}+\ldots+h_tx^{d_t}
~~\mbox{for $h_1,\ldots,h_t \in
  \F$ and $d_1,\ldots,d_t\in\NN$.}
\end{equation}
Sparse polynomials have a compact representation as a sequence of
coefficient-degree pairs $(h_1,d_1),\ldots,(h_t,d_t)$, which allow
representation and manipulation of very high degree polynomials.
Let $f\in\F[x]$ have degree $d$.  We examine the
computation of a $t$-sparse multiple of $f$.  That is, we
wish to determine 
if there exist $g,h\in\F[x]$ such that $fg=h$ and
$h$ has prescribed sparsity $t$,
and if so, to find such an $h$.  We do not attempt to find
$g$, as it may have a super-polynomial number of terms, even though
$h$ has a compact representation
(see Theorem~\ref{thm:bigmul}).

Sparse multiples over finite fields have cryptographic applications.
Their computation is used
in correlation attacks on LFSR-based stream ciphers 
\linebreak
\citep{AG2007,DidLai07}.
The security of the TCHo cryptosystem is also based on the conjectured
computational hardness of sparsest multiple computation over 
$\mathbb{F}_2[x]$ \citep{TCHo07}; our results provide further evidence
that this is in fact a computationally difficult problem.

Sparse multiples can
facilitate efficient arithmetic in extension fields \citep{BreZim03}
and in designing interleavers for error-correcting codes \citep{SadEA01}.
The linear algebra formulation in Section~\ref{sec:la} relates to
finding the minimum distance of a binary
linear code \citep{BerMcE78,Vardy97} as well as
finding ``sparsifications'' of linear systems \citep{EgnerMinkwitz98}.

One of our original motivations was to understand the complexity of
sparse polynomial \emph{implicitization} over $\QQ$ or $\RR$: 
Given a curve represented explicitly as a set of parametric rational
functions, find a sparse polynomial whose zero set contains all points on
the curve (see, e.g., \citet{EmiKot05}). 
This is a useful operation in computer aided geometric design
for facilitating various operations on the curve, and
work here can
be thought of as a univariate version of this problem.

We often consider the related problem of finding a sparse annihilator
for a set of points --- that is, a sparse polynomial with given roots.
This is exactly equivalent to our problem when the input polynomial $f$
is squarefree, and in the binomial case corresponds to asking whether a
given root can be written as a surd. This is also the problem we are
really interested in regarding implicitization, and allows us to build
on significant literature from the number theory community on the roots
of sparse polynomials.

In general, we assume that the desired sparsity $t$ is a constant.
This seems reasonable given that over a finite field, even for
$t=2$,
the problem is probably computationally hard 
(Theorem~\ref{theorem:hardnessoftsparse}).
In fact, we have reason to conjecture that the problem is intractable
over $\QQ$ or $\Fq$ when $t$ is a parameter.
Our algorithms are exponential in $t$ but
polynomial in the other input parameters when $t$ is constant.

Over $\QQ[x]$, the analysis must consider coefficient size,
and we will count machine word
operations in our algorithms to account for coefficient growth. We
follow the conventions of \cite{Len99} and define the \emph{height}
of a polynomial as follows.  Let $f\in\QQ[x]$ and $r\in\QQ$ the
least positive rational number such that $rf\in\ZZ[x]$. If $rf=\sum_i
a_i x^{e_i}$ with each $a_i\in\ZZ$, then the \emph{height} of $f$,
written $\H(f)$, is $\max_i |a_i|$.

We examine variants of the sparse multiple
problem over $\Fq$ and $\QQ$.  
Since every polynomial in $\Fq$ has a 2-sparse multiple of high degree,
given $f\in\Fq[x]$ and $n\in\NN$
we consider the problem of finding a
$t$-sparse multiple of $f$ with degree at most $n$.
For input $f\in\QQ[x]$ of degree
$d$, we consider algorithms which seek $t$-sparse multiples of height
bounded above by an additional input value $c\in\NN$.  We present
algorithms requiring time polynomial in $d$ and $\log c$.

The remainder of the paper is structured as follows.

In Section~\ref{sec:la}, we consider the straightforward linear
algebra formulation of the sparse multiple problem.  This is useful
over $\QQ[x]$ once a bound on the output degree is derived, and also
allows us to bound the output size.  In addition, it connects our
problems with related \NP-complete coding theory problems.

In Section~\ref{sec:rat-bin} we consider the problem of finding the
least-degree binomial multiple of a rational polynomial.  A
polynomial-time algorithm in the size of the input is given which
completely resolves the question in this case.  This works despite the
fact that we show polynomials with binomial multiples whose degrees
and heights are both exponential in the input size!

In Section~\ref{sec:rat-gen} we consider the more general problem of finding a
$t$-sparse multiple of an input $f\in\QQ[x]$.  Given a height bound
$c\in\NN$ we present an algorithm which requires polynomial time in
$\deg f$ and $\log c$, except in the very special case that $f$ 
has both non-cyclotomic and repeated cyclotomic factors.  

Section~\ref{sec:Fq-gen} shows that, even for $t=2$,
finding a $t$-sparse multiple of
a polynomial $f\in\Fq[x]$ is at least as hard
as finding multiplicative orders in an extension of $\Fq$
(a problem thought to be computationally difficult).
This lower bound is shown to be tight for $t=2$ 
due to an algorithm for computing binomial multiples that uses
order finding.

Open questions and avenues for
future research are discussed in Section~\ref{sec:conclusion}.

An extended abstract of some of this work appears in
\cite*{GieRocTil:2010}.  Some of this work and further explorations,
also appears in the Masters thesis of
\cite{Tilak:2010}.

%%% Local Variables: 
%%% mode: latex
%%% TeX-master: "alg_submit"
%%% End: 

%-----------------------------------------------------------------------------%
\section{Linear algebra formulation} 
\label{sec:la}
%-----------------------------------------------------------------------------%

The sparsest multiple problem can be formulated using linear algebra.
This requires specifying bounds on degree, height and sparsity; later
some of these parameters will be otherwise determined. This approach
also highlights the connection to some problems from coding theory.
We exhibit a randomized algorithm for finding a $t$-sparse multiple $h$ of
a degree-$d$ polynomial $f\in\QQ[x]$, given bounds $c$ and $n$ on the
height and degree of the multiple respectively.  
When $t$ is a constant,
the algorithm runs in time polynomial in $n$ and $\log \H(f)$
and returns the desired
output with high probability.  We also conjecture the intractability
of some of these problems, based on similar problems in coding theory.
Finally, we show that the construction of \cite{Vardy97} can be used
to show the problem of finding the sparsest vector in an integer
lattice is NP-complete, which was conjectured by
\cite{EgnerMinkwitz98}.

Let $\R$ be a principal ideal domain, with $f\in\R[x]$ of degree $d$
and $n \in \NN$ given.  Suppose $g,h\in\R[x]$ have degrees $n-d$
and $n$ respectively, with $f=\sum_{0}^{d} f_i x^i$, $g=\sum_{0}^{n-d}
g_ix^i$ and $h=\sum_{0}^{n}h_ix^i$. The coefficients in the equation $fg =
h$ satisfy the following linear system:
\begin{equation}
  \label{eq:linalg}
  \underbrace{\left[\begin{array}{cccc}
        f_0 &&& \\
        f_1 & f_0 && \\
        \vdots & f_1 & \ddots &\\
        f_d & \vdots & \ddots & f_0\\
        & f_d & \ddots & f_1 \\
        && \ddots & \vdots \\
        &&& f_d
      \end{array}\right]}_{A_{f,n}}
  \underbrace{\left[\begin{array}{c}
        g_0 \\ g_1 \\[4mm] \vdots \\[4mm] g_{n-d}
      \end{array}\right]}_{v_g}
  =
  \underbrace{\left[\begin{array}{c}
        h_0 \\ h_1 \\[10mm] \vdots \\[10mm] h_n
      \end{array}\right]}_{v_h}.
\end{equation}
Thus, a multiple of $f$ of degree at most $n$ and sparsity at most $t$
corresponds to a vector with at most $t$ nonzero entries (i.e.,
a $t$-sparse vector) in the linear span of $A_{f,n}$.  

If $f\in\R[x]$ is squarefree and has roots
$\{\alpha_1,\ldots,\alpha_d\}$, possibly over a finite extension of
$\R$, then the following also holds:
\begin{equation}
  \label{eq:rootPowers}
\underbrace{
\left[\begin{array}{cccc}
    1 & \alpha_1 & \cdots & \alpha_1^n\\
    1 & \alpha_2 & \cdots & \alpha_2^n\\
    \vdots & \vdots & \vdots & \vdots\\
    1 & \alpha_d & \cdots & \alpha_d^n\\
\end{array}\right]}_{A_n(\alpha_1,\ldots,\alpha_d)}
\left[\begin{array}{c}
    h_0 \\ h_1 \\[4mm] \vdots \\[4mm] h_n
 \end{array}\right]
= \mathbf{0}.
\end{equation}
Thus $t$-sparse multiples of a squarefree $f$ correspond to $t$-sparse
$\R$-vectors in the nullspace of $A_n(\alpha_1,\ldots,\alpha_d)$.

\subsection{Finding short $l_\infty$ vectors in lattices}
\label{sec:aks}

This technical section presents a randomized, polynomial-time
algorithm to find the shortest $l_\infty$ vector in a
constant-dimensional lattice. Our algorithm is a modification of
\cite{AjtaiKS2001}, based on the presentation by \cite{Reg04}, adapted
to the case of infinity norm.  Since this the techniques are
essentially drawn from the literature, and while necessary, are not
the central thrust of this current paper, full details are left to
Appendix~\ref{app:aks}.

Algorithm~\ref{alg:aksMod} below starts by computing a rough
approximation of the shortest $l_2$ vector using LLL \citep*{LLL82},
and then scales the lattice accordingly.  The main while loop then
consists of two phases: sampling and sieving. First, a large number of
random vectors $\{x_1,\ldots,x_m\}$ are sampled in an
appropriately-sized ball around the origin. We take these modulo the
basis $B$ to obtain vectors $\{y_1,\ldots,y_m\}$ with the property
that each $x_i-y_i$ is in the lattice of $B$. Next, we use a series of
sieving steps in the while loop in Step~\ref{step:sievingStart} to
find a small subset of the $y_i$ vectors that are close to every other
vector and use these as ``pivots''. The pivots are discarded from the
set, but all remaining lattice vectors $x_i-y_i$ are made
smaller. After this, the set $W_\gamma$ contains most lattice vectors
whose $l_2$ length is close to $\gamma$.

\begin{algorithm}[htbp]
  \caption{Shortest $l_\infty$ vector in a lattice}
  \label{alg:aksMod}
  \KwIn{Basis $U\in\ZZ^{n\times d}$ for an integer lattice $\L$ of dimension $n$
  and size $d\leq n$}
  \KwOut{Shortest $l_\infty$ vector in $\L$} 
  $\lambda \gets$ approximate $l_2$-shortest vector in $\L$ from
    \cite{LLL82}\;
  $B \gets (1/\norm{\lambda}_2)\cdot U$, 
  stored as a list of vectors $[b_1,\ldots,b_d]$\;
  \For{$k \in \{1,2,\ldots, 2n\}$ \label{step:aksforstart}}{
    $B \gets 1.5\cdot B$ \label{step:akssetB}\;
    $r_0 \gets n \max_i \norm{b_i}_2$ \;
    $\gamma \gets 3/2$\; 
    \While{$\gamma \leq 3\sqrt{n}+1$ \label{step:akswhile}} { 
      $m \gets \displaystyle\ceil{2^{(7+\ceil{\log \gamma})n} \log r_0}$\; 
      Sample points $\{x_1,\ldots,x_m\}$ uniformly and independently from
        $\ball{0}{\gamma}$, the $n$-dimensional ball of radius $\gamma$
        centered around $\mathbf{0}$\; 
      $S \gets \{1,2,\ldots,m\}$\; 
      $y_i \gets x_i \mod \mathcal{P}(B)$ for every $i \in S$,
        $\mathcal{P}(B)$ being the parallelogram of $B$ defined in the
        proof of Lemma~\ref{lemma:xMinusyIsInLattice}\;
      \label{step:initialEstimateForY} $r \gets r_0$\; 
      \While{$r > 2\gamma + 1$}{
        \label{step:sievingStart}
        $J \gets \emptyset$\; 
        \For{$i \in S$}{ 
          \lIf{$\exists{j\in J}$ such that $\norm{y_j-y_i} \leq r/2$}
            { $\eta_i \gets j$ }\;
          \lElse { $J \gets J \cup \{i\}$ }\; 
        }
        \label{step:sievingEnd} $S \gets S \setminus J$\; 
        $y_i \gets y_i + x_{\eta_i}-y_{\eta_i}$ for $i \in S$\; 
        $r \gets r/2 + \gamma$\;
      }
      $Y_\gamma \gets \{(x_i-y_i) \ |\ i \in S\}$\; 
      $W_\gamma \gets \{v-w \ |\ v,w\in Y_\gamma \text{ and } v\neq w\}$\; 
      $\gamma \gets 3\gamma/2$\; 
    } 
    $v_k \gets$ shortest $l_\infty$ vector in any $W_\gamma$
    \label{step:aksforend}\;
  }
  \KwRet{shortest $l_\infty$ vector in 
  $\{(\norm{\lambda}_2/1.5^k)\cdot v_k\ |\ k=1,2,\ldots,n\}$
  \label{step:aksret}}
\end{algorithm}

If we are fortunate enough that the shortest $l_2$ vector in
the lattice with basis $B$ 
set on Step~\ref{step:akssetB} has length between $2$ and $3$, then we
know that the shortest $l_\infty$ vector in this lattice must have
$l_2$ length between $2$ and $3\sqrt{n}$. By iterating $\gamma$ in the
appropriate range, we will encounter this shortest $l_\infty$ vector and
set it to $v_k$ on Step~\ref{step:aksforend} with high probability. 
We prove, given our approximate starting point from LLL,  we will be in this
``fortunate'' situation in at least one iteration through the outer for
loop.

The correctness and efficiency of the algorithm is given by the
following theorem, whose proof we defer to Appendix~\ref{app:aks}.

\newcommand{\aksworks}{
  Given a lattice basis $U\in\ZZ^{n\times d}$, Algorithm~\ref{alg:aksMod} 
  returns the shortest $l_\infty$ vector in the lattice of $U$,
  with probability at least $1-1/2^{O(n)}$, using 
  $2^{O(n\log n)}\cdot \norm{U}^{O(1)}$ bit operations.
}
\begin{theorem} 
  \label{theorem:aksWorks!} \aksworks
\end{theorem}

\subsection{Finding a sparse multiple of bounded height and \nohyphens{degree}}

We now present an algorithm to find the sparsest bounded-degree,
bounded-height multiple $h\in\QQ[x]$ of an input $f\in\QQ[x]$.  Since
$\H$ is invariant under scaling,
we may assume that $f,g,h\in\ZZ[x]$ .

The basic idea is the following. Having fixed the positions at which
the multiple $h$ has nonzero coefficients, finding a low-height
multiple is reduced to finding the nonzero vector with smallest
$l_\infty$ norm in the image of a small lattice.

Let $I=\{i_1,\ldots,i_t\}$ be a $t$-subset of $\{0,\ldots,n\}$,
and $A^I_{f,n}\in\ZZ^{(n-t+1)\times
(n-d+1)}$ the matrix $A_{f,n}$ with rows $i_1,\ldots,i_t$
removed. Denote by $B^I_{f,n}\in\ZZ^{t\times (n-d+1)}$ the matrix
consisting of the removed rows $i_1,\ldots,i_t$ of the matrix
$A_{f,n}$.  Existence of a $t$-sparse multiple
$h=h_{i_1}x^{i_1}+h_{i_2}x^{i_2}+\cdots+h_{i_t}x^{i_t}$ of input $f$ is equivalent to the existence of a vector $v_g$
such that $A^I_{f,n} \cdot v_g = \mathbf{0}$ and $B^I_{f,n} \cdot v_g =
[h_{i_1},\ldots,h_{i_t}]^T$.

Now let $C^I_{f,n}$ be a matrix whose columns span the nullspace of the matrix $A^I_{f,n}$. Since $A_{f,n}$ has full
column rank, the nullspace of $A^I_{f,n}$ has dimension $s\leq t$ and
$C^I_{f,n} \in \ZZ^{(n-d+1) \times s}$.  Thus, a $t$-sparse
multiple $h=h_{i_1}x^{i_1}+\cdots+h_{i_t}x^{i_t}$ of $f$ exists if
and only if there exists a $v \in \ZZ^s$ such that
\begin{equation}
  \label{eq:linalg-image-vector-equation}
  B^I_{f,n} \cdot C^I_{f,n} \cdot v = [h_{i_1},\ldots,h_{i_t}]^T.
\end{equation}
Note that $B^I_{f,n} \cdot C^I_{f,n}\in\ZZ^{t\times s}$. Our approach,
outlined in Algorithm~\ref{alg:boundednc}, is to generate this lattice
and search for a small, $t$-sparse vector in it. For completeness, we
first define the subset ordering used in the search.

\begin{definition}
  Let $a=(a_1,a_2,\ldots,a_k)$ and $b=(b_1,b_2,\ldots,b_k)$ be two
  $k$-tuples. $a$ precedes $b$ in \emph{reverse lexicographical order}
  if and only if there exists an index $i$ with $1\leq i\leq k$ such
  that $a_i < b_i$, and
  for all $j$ with $i<j\leq k$, $a_j = b_j$.
\end{definition}

\begin{algorithm}[htbp]
  \caption{Bounded-Degree Bounded-Height Sparsest Multiple}
  \label{alg:boundednc}
  \KwIn{$f\in\ZZ[x]$ and $t,n,c\in\NN$}
  \KwOut{A $t$-sparse multiple $h\in\ZZ[x]$ of $f$ with $\deg(h) \leq n$ and $\H(h) \leq c$, or \NONE}
  \For{$s = 2,3,\ldots,t$}{
    \ForEach{ $s$-subset $I=(0,i_2,\ldots,i_s)$ of $\{0,1,\ldots,n\}$,
      \linebreak
      sorted in reverse lexicographic order,
    \label{lamul:subsets}}{
      Compute matrices $A^I_{f,n}$ and $B^I_{f,n}$ as defined above\;
      \If{$A^I_{f,n}$ does not have full column rank}{
        Compute matrix $C^I_{f,n}$, a kernel basis for
        $A^I_{f,n}$\; \label{step:SNF}
        $\mathbf{h} \gets$ shortest $l_\infty$ vector in the 
        lattice of $B^I_{f,n} \cdot C^I_{f,n}$ from
        Algorithm~\ref{alg:aksMod} \; 
        \label{step:latticeShortVector}
        \lIf{$\norm{\mathbf{h}}_\infty \leq c$} {
          \KwRet{$h_1+h_2x^{i_2}+\cdots+h_tx^{i_t}$}\;
        }
      }
    }
  }
  \KwRet{\NONE}
\end{algorithm}

The following lemma shows how to compute Step~\ref{step:SNF}
efficiently using the Smith normal form.
\begin{lemma}
  \label{lem:duallat}
  Given $T\in\ZZ^{k\times \ell}$ with $k\geq\ell$ and nullspace of
  dimension $s$, we can compute a $V\in\ZZ^{\ell\times s}$ such
  that the image of $V$ equals the nullspace of $T$.  The algorithm
  requires $\softO(k\ell^2 s\log\norm{T})$ bit operations 
  (ignoring logarithmic factors).
\end{lemma}
\begin{proof} 
  First compute the Smith normal form of the matrix: $T=PSQ$ for
  diagonal matrix 
  $S=\diag(\delta_1,\ldots,\delta_{\ell-s},0,\ldots,0)\in\ZZ^{k\times\ell}$
  and unimodular matrices $P\in\ZZ^{k\times k}$ and
  $Q\in\ZZ^{\ell\times\ell}$.
  \cite{Sto00} gives efficient algorithms to compute such a $P,S,Q$ with
  $\softO(k\ell^2s\log\norm{T})$ bit operations.

  Then since any vector $\mathbf{v}$ in the nullspace of $T$ satisfies
  $PSQ\mathbf{v}=\mathbf{0}$, $SQ\mathbf{v}=\mathbf{0}$ also and $\mathbf{v}$
  is in the nullspace of $SQ$. 
  Next compute the inverse of $Q$; this can be accomplished with the same
  number of bit operations since $\ell \leq k$. Define $V$ to be the
  last $s$ columns of $Q^{-1}$.
  Due to the diagonal structure of $S$, $V$
  must be a nullspace basis for $SQ$, and
  furthermore $V$ has integer entries since $Q$ is unimodular.
  \QED
\end{proof}

The correctness and efficiency of Algorithm~\ref{alg:boundednc} can then be summarized as follows.
\begin{theorem}
  \label{thm:a1}
  Algorithm~\ref{alg:boundednc} correctly computes a $t$-sparse
  multiple $h$ of $f$ of degree $n$ and height $c$, if it exists, with
  $(\log \H(f))^{O(1)} \cdot n^{O(t)} \cdot 2^{O(t \log{t})}$ bit operations. The sparsity
  $s$ of $h$ is minimal over all multiples with degree less than $n$
  and height less than $c$, and the degree of $h$ is minimal over all such
  $s$-sparse multiples.
\end{theorem}
\begin{proof}
  The total number of iterations of the for loops is $\sum_{s=2}^t
  \binom{n-1}{s-1} < n^t$. Computing the rank of $A^I_{f,n}$, and
  computing the matrices $B^I_{f,n}$ and $C^I_{f,n}$ can each be done
  in polynomial time by Lemma~\ref{lem:duallat}. The size of the
  entries of $C^I_{f,n}$ is bounded by some polynomial
  $(\log{\H(h)}+n)^{O(1)}$. The computation of the shortest $l_\infty$
  vector can be done using $2^{O(t \log{t})}$ operations on numbers of length
  $(\log{\H(h)}+n)^{O(1)}$,
  by Theorem~\ref{theorem:aksWorks!}.

  The minimality of sparsity and degree comes from the ordering of the
  for loops. Specifically, the selection of subsets in Step~\ref{lamul:subsets}
  is performed in \emph{reverse lexicographic order}, so that column subsets
  $I$ corresponding to lower degrees are always searched first.
  \QED
\end{proof}

\subsection{Relationship to NP-hard problems}

Note that the above algorithms require time exponential in $t$, and
are only polynomial-time for constant $t$.  It is natural to ask
whether there are efficient algorithms which require time polynomial
in $t$.  We conjecture this problem is probably
\NP-complete, and point out two results of
\cite{Vardy97} and \cite{GurVar05} on related problems that are known to
be hard.

The formulation \eqref{eq:rootPowers} seeks the sparsest vector in the
nullspace of a (structured) matrix.  For an unstructured matrix over finite fields, this
is the problem of finding the minimum distance of a linear code, 
shown by \cite{Vardy97} to be NP-complete. The same problem over integers translates into finding the sparsest vector in an integer lattice. It was posed as an open problem in \cite{EgnerMinkwitz98}. Techniques similar to \cite{Vardy97} prove that this problem is also \NP-complete over the integers, a fact proved in Theorem~\ref{thm:egnerMinkwitzTheorem}.

Of course, the problem may be easier for structured matrices as in
\eqref{eq:rootPowers} However, \cite{GurVar05} show that maximum
likelihood decoding of cyclic codes, which seeks sparse solutions to
systems of equations of similar structure to \eqref{eq:rootPowers}, is
also NP complete.  They do require the freedom to choose a
right-hand-side vector, whereas we insist on a sparse vector in the
nullspace.  While these two results certainly do not prove that the
bounded-degree sparsest multiple problem is NP-complete, they support
our conjecture that it is.

\begin{theorem}\label{thm:egnerMinkwitzTheorem} The problem \emph{SparseLatticeVector} of computing the vector with the least Hamming weight in an integer lattice specified by its basis is \NP-complete.
\end{theorem}
\begin{proof}
  To see that the problem is in \NP, a nondeterministic machine can just guess the positions at which the lattice vector is nonzero. The rest is a standard linear algebra problem. 

  We now show \NP-hardness by giving a Cook-reduction from the problem \emph{Subset Sum}, a well-known \NP-complete problem.

  We note first the standard formulation of \emph{Subset Sum}: Given
  distinct integers $\{z_1,\ldots,z_n\}$, a target integer $t$ and a positive
  integer $w \leq n$, is there a non-empty subset $S \subseteq
  \{1,\ldots,n\}$ of size exactly $w$ such that such that $\sum_{i \in
  S} z_i = t$?

If $w=n$, the problem can be solved by comparing the sum $\sum_i z_i$ with $t$. Therefore, we can assume that $w < n$. Given an instance $\{z_1,\ldots,z_n\}$ of subset sum, to check if there is a subset of size $w<n$ summing to $t$, the reduction first creates the following matrix:
\begin{equation} \label{egnerMinkwitzMatrix}
  M_w=
  \left[\begin{array}{ccccc}
      1 & 1 & \cdots & 1 & 0\\
      z_1 & z_2 & \cdots & z_n & 0\\
      \vdots   & \vdots  & \vdots & \vdots & \vdots\\
      z_1^{w-1} & z_2^{w-1} & \cdots & z_n^{w-1} & 1\\
      z_1^w & z_2^w & \cdots & z_n^w & t\\
    \end{array}
  \right]
  \in\ZZ^{(w+1) \times (n+1)}.
\end{equation}

Lemma~\ref{egnerMinkwitzLemma} (stated and proved below) shows that $M_w$
has a null vector of sparsity at most $w+1$ if and only if
$z_{i_1}+z_{i_2}+\cdots+z_{i_w}=t$ for some $i_1<i_2<\ldots<i_w$.

To create an instance of \emph{SparseLatticeVector}, the reduction creates a matrix $N$ such that the
  columns of $N$ span the kernel of $M$ via $\ZZ$-linear combinations (see Lemma~\ref{lem:duallat}).  The instance
  $(\L,w)$, where $\L$ is the column lattice $\L$ of $N$, is fed to an
  algorithm claiming to solve the \emph{Sparse Vector Problem}.
  \QED
\end{proof}

\begin{lemma}
  \label{egnerMinkwitzLemma} 
  The matrix $M_w$ from equation~\eqref{egnerMinkwitzMatrix} has a
  null vector of Hamming weight $w+1$ if and only if
  $z_{i_1}+z_{i_2}+\cdots+z_{i_w}=t$ for some $i_1<i_2<\ldots<i_w$.
\end{lemma}
\begin{proof}
%  The row-rank of $M_w$ is $w+1$ since $w < n$ and the principal $(w+1)\times(w+1)$ minor is Vandermonde.
  We will first prove that the sparsest null vector has weight at
  least $(w+1)$. To see this, consider the submatrix formed by any set
  of $w$ columns. (We can assume that the last column is included in
  this set since otherwise the submatrix has a Vandermonde minor of
  size $w \times w$, and hence the columns are independent.) Since the
  principal minor of such a submatrix is a $(w-1) \times (w-1)$-sized
  Vandermonde matrix, the rows are independent. On adding either of
  the last two rows, the row-rank only increases since the other rows
  do not contain a nonzero entry in the last coordinate. Hence the
  row-rank (and hence the column-rank) of this submatrix is at least
  $w$, and hence the sparsest null vector of $M_w$ has weight at least
  $(w+1)$.

  Consider a $(w+1)$-sized subset of columns. If the
  last column is not in this set, the chosen columns form a
  Vandermonde matrix with nonzero determinant (since $z_i$ are distinct). Therefore assume that the last column is among those chosen, the determinant
  of the resulting matrix can be expanded as:
  \begin{equation*}
    \left |\begin{array}{ccccc}
        1 & \cdots & 1 & 0\\
        z_{i_1} & \cdots & z_{i_w} & 0\\
        \vdots   &    & \vdots & \vdots\\
        z_{i_1}^{w-1} & \cdots & z_{i_w}^{w-1} & 1\\
        z_{i_1}^w & \cdots & z_{i_w}^w & t\\
      \end{array}\right |=
    t \left |\begin{array}{cccc}
        1 & \cdots & 1\\
        z_{i_1} & \cdots & z_{i_w}\\
        \vdots   & \vdots & \vdots\\
        z_{i_1}^{w-1} & \cdots & z_{i_w}^{w-1}\\
      \end{array}\right |
    - \left |\begin{array}{cccc}
        1 & \cdots & 1\\
        z_{i_1} & \cdots & z_{i_w}\\
        \vdots   & \vdots & \vdots\\
        z_{i_1}^{w-2} & \cdots & z_{i_w}^{w-2}\\
        z_{i_1}^{w} & \cdots & z_{i_w}^{w}\\
      \end{array}\right |.
  \end{equation*}
  The first of the matrices on the right-hand side is a Vandermonde
  matrix, whose determinant is well-known to be $\prod_{i_j <
    i_k}(z_{i_k}-z_{i_j})$. The second matrix is a first-order
  alternant whose determinant is known to be $(z_{i_1}+z_{i_2}+\cdots+
  z_{i_w})\prod_{i_j < i_k}(z_{i_k}-z_{i_j})$.  Hence the
  determinant of the entire matrix is $(t-z_{i_1}-z_{i_2}-\cdots-
  z_{i_w})\prod_{i_j < i_k}(z_{i_k}-z_{i_j})$. Since all the $z_i$
  are distinct, the determinant vanishes if and only if the first term
  vanishes which holds when there exists a subset of
  $\{z_1,z_2,\ldots,z_n\}$ of size $w$ summing to $t$.  
  \QED
\end{proof}

%%% Local Variables:
%%% mode: latex
%%% TeX-master: "alg_submit"
%%% End:

%-----------------------------------------------------------------------------%
\section{Binomial multiples over $\QQ$} 
\label{sec:rat-bin}
%-----------------------------------------------------------------------------%

In this section we completely solve the problem of determining if
there exists a binomial multiple of a rational input polynomial (i.e.,
a multiple of sparsity $t=2$).  That is, given input $f\in\QQ[x]$ of degree $d$,
we determine if there exists a binomial multiple $h=x^m-a\in\QQ[x]$ of
$f$, and if so, find such an $h$ with minimal degree.  The constant
coefficient $a$ will be given as a pair $(r,e)\in\QQ\times\NN$
representing $r^e\in\QQ$.  The algorithm requires a number of bit
operations which is polynomial in $d$ and $\log\H(f)$.  No a priori
bounds on the degree or height of $h$ are required.  We show that $m$
may be exponential in $d$, and $\log{a}$ may be exponential in $\log{\H(f)}$, and
give a family of polynomials with these properties.

\begin{algorithm}[htbp]
  \caption{Lowest degree Binomial Multiple of a Rational Polynomial}
  \label{alg:binomialRationals}
  \KwIn{$f\in\QQ[x]$}
  \KwOut{The lowest degree binomial multiple $h\in\QQ[x]$ of $f$, or \NONE}
  Factor $f$ into irreducible factors:  $f=x^bf_1 f_2 \cdots f_u $\; \label{step:bin-rat-factor-step}
  \lIf{$f$ is not squarefree} {
    \KwRet{\NONE}
  }\;
  \For{$i = 1,2,3,\ldots,u$} {
    $d_i \gets \deg f_i$\;
    $m_i \gets$ least $k\in\{d_i,d_i+1,\ldots,d_i\cdot(\ceil{3d_i\ln\ln d_i}+7)\}$
    such that  $x^k \rem f_i \in \QQ$\; \label{step:bin-mult-for-irred-rat}
    \lIf{no such $m_i$ is found}{\KwRet{\NONE}}\;
    \lElse{$r_i \gets x^{m_i} \rem f_i$}
  }

  $m \gets \lcm(m_1,\ldots,m_u)$\;  

  \ForEach{2-subset $\{i,j\} \subseteq \{1,\ldots,u\}$}{
    \lIf{$\abs{r_i}^{m_j} \neq \abs{r_j}^{m_i}$} {
      \KwRet{\NONE}
    }\;
    \lElseIf{$\mathop{\rm sign}(r_i^{m/m_i}) \neq \mathop{\rm sign}(r_j^{m/m_j})$} {
      $m \gets 2 \cdot \lcm(m_1,\ldots,m_u)$\;
    }
  }
  
  \KwRet{$x^b(x^m-r_1^{m/m_1})$,} {\it with $r_1$ and $m/m_1$ given
    separately}
\end{algorithm}

Algorithm~\ref{alg:binomialRationals} begins by factoring the given
polynomial $f \in \QQ[x]$ into irreducible factors (using, e.g., the
algorithm of \cite{LLL82}).  
We then show how to find a binomial multiple of each irreducible
factor, and finally
provide a combining strategy for the different multiples.

The following theorem of \cite{Risman76} characterizes binomial
multiples of irreducible polynomials.  Let $\phi(n)$ be Euler's totient
function, the number of positive integers less than or equal to $n$
which are coprime to $n$.

\begin{fact}[\protect\cite{Risman76}, Proposition 4, Corollary 2.2]
  \label{fact:risman}
  Let $f\in\QQ[x]$ be irreducible of degree $d$. Suppose the
  least-degree binomial multiple of $f$ (if one exists) is of degree
  $m$. Then there exist $n,t\in\NN$ with
  $n\divs d$ and $\phi(t)\divs d$ such that
  $m=n\cdot t$.
\end{fact}

The following, easily derived from explicit bounds in
\cite{RosSch62}, gives a polynomial bound on $m$.

\begin{lemma} \label{lem:totient-bound}
  For all integers $n\geq 2$, $\phi(\ceil{3n\ln\ln n}+7)>n$.
\end{lemma}
\begin{proof}
   \cite{RosSch62}, Theorem 15, implies that for all $n\geq 3$
   \[
   \phi(n)> \frac{0.56146\cdot n}{\ln\ln n + 1.40722}.
   \]
   It is then easily derived by basic calculus that
   \[
     \phi(3n\log\log n)>\frac{0.56146\cdot (3n\log\log n)}{\ln\ln (3n\log\log n) + 1.40722}>n
   \]
   for $n\geq 24348$. The inequality in the lemma statement is verified
   mechanically (say using Maple) for $2\leq n\leq 24348$.  
   \QED
\end{proof}

Combining Fact~\ref{fact:risman} with Lemma~\ref{lem:totient-bound},
we obtain the following explicit upper bound on the
maximum degree of a binomial multiple of an irreducible polynomial.

\begin{theorem} 
  \label{thm:binomdegbound}
  Let $f\in\QQ[x]$ be irreducible of degree $d$. If a binomial
  multiple of $f$ exists, and has minimal degree $m$, then $m \leq
  d\cdot (\ceil{3d\ln\ln d}+7)$.
\end{theorem}
\begin{proof}
  By Fact~\ref{fact:risman}, $m=n\cdot t$ such that $n\divs d$ and
  $\phi(t)\divs d$.  Define $\xi(n)=\ceil{3n\ln\ln n}+7$, and define
  $\xi^{-1}(n)$ to be the smallest integer such that
  $\xi(\xi^{-1}(n))\geq n$.  From Lemma~\ref{lem:totient-bound}, we
  have that $\phi(\xi(n)) > n$ for $n\geq 2$.  Hence, $d \geq \phi(t)
  \geq \xi^{-1}(t)$. Since $\xi$ is a non-decreasing function, $d \geq
  \xi^{-1}(t)$ implies that $\xi(d) \geq t$.  Thus $m=n \cdot t \leq
  d \cdot \xi(d) \leq d\cdot(\ceil{3d\ln\ln d}+7)$.  
  \QED
\end{proof}

The above theorem ensures that for an irreducible $f_i$, Step~\ref{step:bin-mult-for-irred-rat} of Algorithm~\ref{alg:binomialRationals}
computes the least-degree
binomial multiple $x^{m_i}-r_i$ if it exists, and otherwise
correctly reports failure.
It clearly runs in polynomial time.

If $f$ has any repeated factor, then it cannot have a binomial multiple
(see Lemma~\ref{lem:multbound} below).
So assume the factorization of $f$ is as computed in
Step~\ref{step:bin-rat-factor-step}, and moreover $f$ is squarefree.
If any factor does not have a binomial multiple, neither can
the product. If every irreducible factor does have a binomial multiple, Step~\ref{step:bin-mult-for-irred-rat} computes the one with the least degree.
The following relates the degree of the minimal
binomial multiple of the input polynomial to those of its irreducible factors.

\begin{lemma}
  \label{lem:bimultdeg}
  Let $f\in\QQ[x]$ be such that $f=f_1\cdots
  f_u\in\QQ[x]$ for distinct, irreducible $f_1,\ldots,f_u\in\QQ[x]$.
  Let $f_i\divs (x^{m_i}-r_i)$ for minimal $m_i\in\NN$ and $r_i\in\QQ$,
  and let $f \divs (x^m - r)$ for $r \in \QQ$.
  Then $\lcm(m_1,\ldots,m_u)\divs m$.
\end{lemma}
\begin{proof}
  It suffices to prove that if $f \divs (x^m-r)$ and $f_i \divs
  (x^{m_i}-r_i)$ for minimal $m_i$ then $m_i \divs m$ since any multiple of $f$ is also a multiple of $f_i$.

  Assume for the sake of contradiction that $m = cm_i + \ell$ for 
  $0 < \ell < m_i$. Then for any root $\alpha_i \in \CC$ of $f_i$, we have that 
  $r=\alpha^{m}=\alpha^{cm_i}\cdot \alpha^{\ell}=r_i^c \cdot \alpha^\ell$. Since $r$ and $r_i$ are both rational, so is $\alpha^\ell$. Also $\alpha^\ell=\beta^\ell$ for any two roots $\alpha,\beta \in \CC$ of $f_i$. Hence $f_i \divs x^\ell-\alpha^\ell$ and $\ell < m_i$, contradicting the minimality of $m_i$.

  Thus $m_i \divs m$, and therefore $\lcm(m_1,\ldots,m_u) \divs m$.
  \QED
\end{proof}

\begin{lemma}
  \label{lemma:rat-bin-combine-step} 
  For a polynomial $f \in \QQ[x]$ factored into distinct irreducible
  factors $f=f_1f_2\ldots f_u$, with $f_i \divs (x^{m_i}-r_i)$ for $r_i
  \in \QQ$ and minimal such $m_i$, a binomial multiple of $f$ exists
  if and only if $\abs{r_i}^{m_j} = \abs{r_j}^{m_i}$ for every pair $1\leq
  i,j\leq u$. If a binomial multiple exists, the least-degree binomial
  multiple of $f$ is $x^{m}-r_i^{m/m_i}$ such that $m$ either equals
  the least common multiple of the $m_i$ or twice that number. It can
  be efficiently checked which of these cases holds.
\end{lemma}
\begin{proof} 
  Let $\alpha_i \in \CC$ be a root of $f_i$. For any candidate binomial multiple $x^m-r$ of $f$, we have (from Lemma~\ref{lem:bimultdeg}) that $m_i \divs m$.

  First, suppose that such a binomial multiple exists: $f \divs (x^m-r)$ with $r \in \QQ$. It is easily seen from $\alpha_i^m=r$ and $\alpha_i^{m_i}=r_i$ that $r_i^{m/m_i}=r$. Since this holds for any $f_i$, we see that $r_i^{m/m_i}=r=r_j^{m/m_j}$ for any $1 \leq i, j \leq u$. Thus $\abs{r_i}^{m_j}=\abs{r_j}^{m_i}$ must hold.

  Conversely, suppose that $\abs{r_i}^{m_j}=\abs{r_j}^{m_i}$ holds for
  every pair $1 \leq i,j \leq u$. We get that $\abs{\alpha_i}^{\ell
  m_im_j}=\abs{\alpha_j}^{\ell m_jm_i}$, and hence $\abs{\alpha_i^\ell
  }=\abs{\alpha_j^\ell }$ for \linebreak 
  $\ell =\lcm(m_1,\ldots,m_u)$. But
  $\alpha_i^\ell $ are all rational since $m_i \divs \ell $. Thus
  $\alpha_i^{2\ell }=\alpha_j^{2\ell }$ for every pair $i,j$ . Thus,
  there exists a binomial multiple of the original polynomial of degree
  $2\ell $.

  To check whether $\alpha_i^\ell =\alpha_j^\ell $ holds (or in other
  words if the degree of the binomial multiple is actually the lcm), it
  suffices to check whether the sign of each $\alpha_i^\ell $ is the
  same. This is equivalent to checking whether the sign of each
  $r_i^{\ell /m_i}$ is the same. Since we can explicitly compute $\ell $
  and all the $r_i$, the sign of each $r_i^{\ell /m_i}$ can be easily
  computed from the sign of $r_i$ and the parity of $\ell /m_i$.
  \QED
\end{proof}

The following comes directly from the previous lemma and the fact that
Algorithm~\ref{alg:binomialRationals} performs polynomially many
arithmetic operations.

\begin{theorem} 
  \label{thm:correctness-and-efficiency-ratbin} 
  Given a polynomial $f \in \QQ[x]$,
  Algorithm~\ref{alg:binomialRationals} outputs the least-degree
  binomial multiple $x^m-r_i^{m/m_i}$ (with $r_i$ and $m/m_i$ output
  separately) if one exists or correctly reports the lack of a
  binomial multiple otherwise. Furthermore, it runs in deterministic
  time $(d+\H(f))^{O(1)}$.
\end{theorem}

The constant coefficient of the binomial multiple cannot be output in
standard form, but must remain an unevaluated power; the next
theorem exhibits an infinite family of polynomials whose minimal
binomial multiples have exponentially sized degrees and heights.

\begin{theorem}
  \label{thm:bigmul}
  For any $d\geq 841$ there exists a polynomial $f\in\ZZ[x]$ of
  degree at most $d\log d$ and height $\H(f)\leq \exp(2d\log d)$
  whose minimal binomial multiple $x^m-a$ is such
  that $m>\exp(\sqrt{d})$ and $\H(a)>2^{\exp(\sqrt{d})}$.
\end{theorem}
\begin{proof}
  We construct the family from a product of cyclotomic polynomials.
  Let $p_i\in\NN$ be the $i$\textsuperscript{th} largest prime, and let
  $\Phi_{p_i}=(x^{p_i}-1)/(x-1)\in\ZZ[x]$ be the $p_i$\textsuperscript{th} cyclotomic
  polynomials (whose roots are the primitive $p_i$\textsuperscript{th} roots of unity).
  This is well known to be irreducible in $\QQ[x]$.

  Let $\ell=\sqrt{2d}$ and $g=\prod_{1\leq i\leq\ell} \Phi_{p_i}$.
  Then, using the fact easily derived from \cite{RosSch62}, Theorem 3, that $i\log
  i<p_i<1.25 i\log i$ for all $i\geq 25$ and verifying that $(p_i - 1) \leq 1.5 i \log{i}$ mechanically for smaller values of $i$,
  \[
  \deg g=\sum_{1\leq i\leq\ell} (p_i-1)\geq \sum_{1\leq i\leq\ell}
  i = \frac{l(l+1)}{2} \geq d,
  \]
  and
  \[
  \deg g = \sum_{1\leq i\leq\ell} (p_i-1) \leq \sum_{1\leq i\leq\ell}
  1.5i\log i \leq 1.5\left(\frac{\ell^2+\ell}{2}\log\ell\right)
  \leq d\log d.
  \]
   The degree $m$ of the minimal binomial multiple is the
  lcm of the order of the roots, and hence equal to the product of
  primes less than or equal to $p_\ell$. This is
  $\exp(\vartheta(p_\ell))$ (where $\vartheta$ is the Chebyshev theta
  function), and for $\ell\geq 41$
  \[
  m\geq \exp(\vartheta(p_\ell))\geq \exp(\vartheta(\ell)) 
  \geq\exp\left(\ell \left(1-\frac{1}{\log\ell}\right)\right)
  \geq \exp\left(\sqrt{d}\right),
  \]
  for $d\geq 841$, where the bounds on $\vartheta$ are derived from
  \cite{RosSch62} Theorem 4.

  Now let $f=g(2x)$, so the minimal binomial multiple of $f$ is
  $x^m-1/2^m$.  We have that
  \[
  \H(g)\leq \prod_{1\leq i\leq\ell}(1+p_i) \leq 2^\ell \prod_{1\leq
    i\leq\ell} p_i
  \leq \exp(2\ell\log\ell)
  \]
  and
  \[
  \H(f)\leq 2^{\deg(g)}\H(g)
  \leq 2^{d\log d}\exp(d\log d+2\sqrt{2d}\log\sqrt{2d})
  \leq \exp(2d\log d)
  \]
  for all $\geq 841$. 
  \QED
\end{proof}

%%% Local Variables: 
%%% mode: latex
%%% TeX-master: "alg_submit"
%%% End: 

%-----------------------------------------------------------------------------%
\section{Computing $t$-sparse multiples over $\QQ$} 
\label{sec:rat-gen}
%-----------------------------------------------------------------------------%

We examine the problem of computing $t$-sparse multiples of
rational polynomials, for any fixed positive integer $t$. As with other
types of polynomial computations, it seems that cyclotomic polynomials
behave quite differently from cyclotomic-free ones. Accordingly, we first
examine the case that our input polynomial $f$ consists only of
cyclotomic or cyclotomic-free factors. Then we see how to combine them, in the
case that none of the cyclotomic factors are repeated.

Specifically, we will show that, given any rational polynomial $f$
which does not have repeated cyclotomic factors, and a height bound
$c\in\NN$, we can compute a sparsest multiple of $f$ with height at
most $c$, or conclude that none exists, in time polynomial in the size
of $f$ and $\log c$ (but exponential in $t$).

First, notice that multiplying a polynomial by a power of $x$ does not
affect the sparsity, and so without loss of generality we may assume all
polynomials are relatively prime to $x$; we call such polynomials
\emph{non-original} since they do not pass through the origin.

\subsection{The cyclotomic case}

Suppose the input polynomial $f$ is a product of cyclotomic factors, and
write the complete factorization of $f$ as
\begin{equation}
\label{eqn:cycfac}
f=\Phi_{i_1}^{e_i}\cdot \Phi_{i_2}^{e_2} \cdots \Phi_{i_k}^{e_k},
\end{equation}
where $\Phi_j$ indicates the $j^\textrm{th}$ cyclotomic polynomial,
the $i_j$'s are all distinct, and the $e_i$'s are positive integers.

Now let $m = \lcm(i_1,\ldots,i_k)$. Then $m$ is the least integer such that
$\Phi_{i_1}\cdots\Phi_{i_k}$ divides $x^m-1$. Let
$\ell=\max_i e_i$, the maximum multiplicity of any factor of $f$. This means
that $(x^m-1)^\ell$ is an $(\ell+1)$-sparse multiple of $f$. To prove that
this is in fact a sparsest multiple of $f$, we first require the following
simple lemma.
Here and for the remainder, for a univariate polynomial $f\in\F[x]$,
we denote by $f'$ the first derivative with respect to $x$, that is,
$\frac{\rm d}{{\rm d}x}f$.

\begin{lemma}\label{lem:multbound}
  Let $h\in\QQ[x]$ be a $t$-sparse and non-original polynomial,
  and write
  $h=a_1 + a_2x^{d_2} + \cdots + a_tx^{d_t}$.
  Assume the complete factorization of $h$ over $\QQ[x]$ is
  $h = a_t h_1^{e_1} \cdots h_k^{e_k}$, 
  with each $h_i$ monic and irreducible.
  Then $\max_i e_i \leq t-1$.
\end{lemma}
\begin{proof}
  Without loss of generality, assume $h$ is \emph{exactly} $t$-sparse, and each $a_i\neq 0$.

  The proof is by induction on $t$. If $t=1$ then $h=a_1$ is a constant, so
  $\max_i e_i=0$ and the statement holds. Otherwise, assume
  the statement holds for $(t-1)$-sparse polynomials.
  
  Write the so-called ``sparse derivative'' $\tilde{h}$ of $h$ as
  \[
  \tilde{h} = \frac{h'}{x^{d_2-1}} =
    a_2 d_2  + a_3 d_3 x^{d_3-d_2} + \cdots +
    a_{t-1}d_{t-1}x^{d_{t-1}-d_2}.
    \] 
  For any $i$ with $e_i>0$, we know that $h_i^{e_i-1}$ divides
  $\frac{d}{dx} h$, and $h_i$ is relatively prime to $x^{d_2-1}$ since
  the constant coefficient of $h$ is nonzero. Therefore $h_i^{e_i-1}$
  divides $\tilde{h}$.
  By the inductive hypothesis, since $\tilde{h}$ is $(t-1)$-sparse and 
  non-original, $e_i-1 \leq t-2$, and therefore
  $e_i \leq t-1$. Since $i$ was chosen arbitrarily, $\max_i e_i \leq
  t-1$.
  \QED
\end{proof}

An immediate consequence is the following:

\begin{corollary}\label{cor:cyclmul}
  Let $f\in\QQ[x]$ be a product of cyclotomic polynomials, written as in
  \eqref{eqn:cycfac}. Then $$h=(x^{\lcm(i_1,\ldots,i_k)}-1)^{\max_i e_i}$$
  is a sparsest multiple of $f$.
\end{corollary}
\begin{proof}
  Clearly $h$ is a multiple of $f$ with exactly $\max_i e_i+1$ nonzero terms.
  By way of contradiction, suppose a $(\max_i e_i)$-sparse multiple of 
  $f$ exists; call it $\bar{h}$. 
  Without loss of generality, we can assume that $\bar{h}$
  is non-original. Then from 
  Lemma~\ref{lem:multbound}, the maximum multiplicity of any factor of $\bar{h}$
  is $\max_i e_i - 1$. But this contradicts the fact that each $\Phi_i^{e_i}$
  must divide $\bar{h}$. Therefore the original statement is false, and
  every multiple of $f$ has at least $\max_i e_i + 1$ nonzero terms.
  \QED
\end{proof}

\subsection{The cyclotomic-free case}

We say a polynomial $f\in\QQ[x]$ is \emph{cyclotomic-free} if it contains
no cyclotomic factors. Here we will show that a sparsest multiple of a
cyclotomic-free polynomial must have degree bounded by a polynomial
in the size of the input and output.

First we need the following elementary
lemma. 

\begin{lemma}\label{lem:deriv}
  Suppose $f,h\in\QQ[x]$ with $f$ irreducible, and $k$ is a positive integer.
  Then $f^k|h$ if and only if $f|h$ and $f^{k-1}|h'$.
\end{lemma}
\begin{proof}
  The $\Rightarrow$ direction is straightforward.

  For the $\Leftarrow$ direction, suppose $f|h$ and $f^{k-1}|h'$.
  Let $\ell$ be the maximum multiplicity of $f$ in $h$, and write
  $h=f^\ell g$ with $g\in\QQ[x]$ relatively prime to $f$.

  We can write $h' = f^{\ell-1}\left( fg' + \ell f' g\right)$. Now, by way of
  contradiction, assume that $k > \ell$. Then $f$ divides $fg' + \ell f'g$, and
  therefore $f$ divides $\ell f' g$. But this is impossible from the assumption
  that $f$ is irreducible and relatively prime to $g$. Therefore $k \leq \ell$,
  and $f^k | f^\ell | h$.
  \QED
\end{proof}

The following technical lemma
provides the basis for our degree bound on the
sparsest multiple of a non-cyclotomic polynomial.

\begin{lemma}\label{lem:gap}
  Let $f,h_1,h_2,\ldots,h_\ell \in\QQ[x]$ be non-original polynomials,
  where $f$ is irreducible and non-cyclotomic with degree $d$,
  and each $h_i$ satisfies 
  $\deg h_i \leq u$ and $\H(h_i)\leq c$.
  Also let $k,m_1,m_2,\ldots,m_\ell$ be positive integers such that
  \[
  f^k|(h_1 x^{m_1} + h_2 x^{m_2} + \cdots + h_\ell x^{m_\ell}).
  \]
  Then $f^k$ divides each $h_i$ whenever every ``gap length'', for
  $1\leq i<\ell$, satisfies
  \begin{equation}\label{eqn:gap1}
    m_{i+1} - m_i - \deg h_i \geq
    \frac{1}{2} d \cdot \ln^3 (3d) \cdot 
    \ln\left(u^{k-1} c \left(t-1\right)\right).
  \end{equation}
\end{lemma}
\begin{proof}
  The proof is by induction on $k$. For the base case, let $k=1$.
  Then we have a separate, inner induction on $\ell$. The inner base
  case, when $k=\ell=1$, is clear since $f$ is non-original. Now assume
  the lemma holds whenever $k=1$ and $1\leq\ell-1<r$ for some $r\geq 2$.
  Let $g_1=h_1x^{m_1}$ and $g_2 = h_2 + \cdots + h_\ell x^{m_r-m_2}$,
  so that $f \divs (g_1 + g_2 x^{m_2})$.
  Since $$m_2-\deg g_1 \geq \frac{1}{2}d\cdot \ln^3 (3d)\cdot \ln(c(t-1)),$$
  we can apply \cite[Proposition 2.3]{Len99} to conclude that
  $f\divs g_1$ and $f\divs g_2$. This means $f\divs h_1$ and, by the inner
  induction hypothesis, $f\divs h_i$ for $2\leq i\leq \ell$ as well. Therefore
  the lemma holds whenever $k=1$.

  Now assume the lemma holds whenever $\ell\geq 1$ and $1\leq k < s$, for
  some $s \geq 2$. Next let $\ell$ be arbitrary and $k=s$. So we write
  $f^s|(h_1 x^{m_1} + \cdots + h_\ell x^{m_\ell})$.
  
  The derivative of the right hand side is
  $$h_1' x^{m_1} + m_1 h_1 x^{m_1 - 1} + \cdots + h_\ell' x^{m_\ell} + 
    m_\ell h_\ell x^{m_\ell-1},$$ which must be divisible by $f^{s-1}$.
  But by the induction
  hypothesis, $f^{s-1}$ also divides each $h_i$, so we can remove all terms with
  $h_i$ from the previous formula and conclude that
  $f^{s-1}|\left( h_1'x^{m_1} + \cdots + h_\ell'x^{m_\ell} \right)$.

  Since each $\H(h_i)\leq c$ and $\deg h_i \leq u$, the height of
  the derivative satisfies $\H(h_i') \leq uc$. A second application of the 
  induction hypothesis therefore shows that each $h_i'$ is divisible by
  $f^{s-1}$. Since $s-1\geq 1$, we already know that each $h_i$ is divisible
  by $f$, and then applying Lemma~\ref{lem:deriv} completes the proof.
  \QED
\end{proof}

Our main tool in proving that 
Algorithm~\ref{alg:boundednc} is useful for
computing the sparsest multiple of a rational polynomial,
given only a bound $c$ on the height, in polynomial time in 
the size of $f$ and $\log c$,
is the following 
degree bound on the sparsest height-bounded
multiple of a rational polynomial. 

\begin{theorem} \label{thm:degbound}
  Let $f\in\QQ[x]$ with $\deg f=d$ be cyclotomic-free,
  and let $t,c\in\NN$ such that $f$ has a nonzero $t$-sparse multiple with
  height at most $c$.
  Denote by $n$ the smallest degree of any such multiple of $f$.
  Then $n$ satisfies
  \begin{equation}\label{eqn:degbound}
    n \leq 2 (t-1) B \ln B,
  \end{equation}
  where $B$ is the formula polynomially bounded by $d$, $\log c$, and $\log t$
  defined as
  \begin{equation}\label{eqn:B}
    B = \frac{1}{2} d^2 \cdot \ln^3(3d) \cdot 
    \ln\left( \hat{c} \left(t-1\right)^d \right),
  \end{equation}
  and $\hat{c} = \max(c,35)$.
\end{theorem}
\begin{proof}
  Let $h$ be a $t$-sparse multiple of $f$ with degree $n$ and height
  $\H(h)\leq c$.
  Without loss of generality, assume $d\geq 1$, $t \geq 2$, and both $f$
  and $h$ are non-original.

  By way of contradiction, assume $n > 2 (t-1)B\ln B$.
  For any univariate polynomial define the \emph{gap lengths} to
  be the differences of consecutive exponents of nonzero terms.
  Split $h$ at every gap greater than $2B\ln B$ by writing
  $$h = h_1 x^{m_1} + h_2 x^{m_2} + \cdots + h_\ell x^{m_\ell},$$
  where each $h_i\in\QQ[x]$ has nonzero constant term and each gap length
  satisfies $m_{i+1}-m_i-\deg h_i > 2B\ln B$. Since we split $h$ at
  \emph{every} sufficiently large
  gap, and $h$ has at most $t$ nonzero terms, 
  each $h_i$ has degree at most $u=2(t-1)B\ln B$.

  We want to show that the gap length $2B\ln B$ is sufficiently large
  to apply Lemma~\ref{lem:gap}. For this, first notice that
  $2B\ln B = B \ln (B^2)$. Since $B$ is positive, $B^2 > 2B\ln B$, so
  the gap length is greater than $B\ln(2B\ln B)$.

  Since $\hat{c}\geq 35$, $B \geq 2.357$, and then
  \begin{align*}
  (d-1)\ln(2B\ln B) \cdot \ln(\hat{c}(t-1)^d) & > 
    \ln\left(\left(2B\ln B\right)^{d-1} \cdot \hat{c} (t-1)^d\right)
    \\
    & =  \ln\left(u^{d-1} \hat{c} \left(t-1\right)\right).
  \end{align*}

  Then from the definition of $B$ in \eqref{eqn:B}, the gap length 
  satisfies
  $$ 2B\ln B > B\ln(2B\ln B) > 
  \frac{1}{2} d \cdot \ln^3(3d) \cdot \ln\left(u^{d-1}\hat{c}\left(t-1\right)\right).
  $$

  Finally, notice that the maximum multiplicity of any factor of $f$ is at most
  $\deg f = d$.  Thus, using the notation of Lemma~\ref{lem:gap}, $d \geq k$.
  Therefore Lemma~\ref{lem:gap} applies to each factor of $f$ (to full
  multiplicity) and we conclude that $f$ divides each $h_i$.

  But then, since there is at least one gap and $\ell > 1$, $h_1$ is a
  multiple of $f$ with fewer terms and lower degree than $h$. This is a
  contradiction, which completes the proof.
  \QED
\end{proof}

In order to compute the sparsest multiple of a rational polynomial with
no cyclotomic or repeated factors, we therefore can simply call 
Algorithm~\ref{alg:boundednc} with the given height bound $c$
and degree bound as specified in
\eqref{eqn:degbound}. 

\subsection{Handling cyclotomic factors}

Suppose $f$ is any non-original rational polynomial with no repeated
cyclotomic factors. Factor $f$ as
$f = f_C \cdot f_D$, where $f_C$ is a squarefree product of cyclotomics
and $f_D$ is cyclotomic-free.
Write the factorization of $f_C$ as
$f_C = \Phi_{i_1} \cdots \Phi_{i_k}$, where $\Phi_n$ is
the $n$\textsuperscript{th} cyclotomic polynomial. Since every $i$\textsuperscript{th} root of unity is also
a $(mi)$\textsuperscript{th} root of unity for any $m\in\NN$, $f_C$ must divide the binomial
$x^{\lcm\{i_1,\ldots,i_k\}}-1$, which is in fact a sparsest multiple of
$f_C$ (Corollary \ref{cor:cyclmul}) and clearly has minimal height.

Then we will show that
a sparsest height-bounded multiple of $f$ is either of small degree, or 
can be constructed as
a sparsest height-bounded multiple of $f_D$ times 
the binomial multiple of $f_C$ specified above. Algorithm~\ref{alg:sqfree}
uses this fact to compute a sparsest multiple of any such $f$.

\begin{algorithm}[htbp]
\caption{Rational Sparsest Multiple}
\label{alg:sqfree}
\KwIn{Bounds $t,c\in\NN$ and $f\in\QQ[x]$ a non-original polynomial of
degree $d$ with no repeated cyclotomic factors}
\KwOut{$t$-sparse multiple $h$ of $f$ with
$\H(h)\leq c$, or \NONE}
  Factor $f$ as $f = \Phi_{i_1}\cdot \Phi_{i_2}\cdots \Phi_{i_k} \cdot f_D$,
    where $f_D$ is cyclotomic-free\label{alg:sqfree:factor}\;
  $n \gets$ degree bound from \eqref{eqn:degbound}\;
  $\hat{h} \gets$ $\floor{t/2}$-sparse 
  multiple of $f_D$ with $\H(\hat{h})\leq c$ and
  $\deg \hat{h} \leq n$, using Algorithm~\ref{alg:boundednc}\;
  $\tilde{h} \gets$ $t$-sparse multiple of $f$ with $\H(h)\leq c$ and
  $\deg h\leq n$, using Algorithm~\ref{alg:boundednc}\;
  \lIf{$\hat{h}=$\NONE and $\tilde{h}=$\NONE}
  {\KwRet{\NONE}}\;
  \lElseIf{$\hat{h}=$\NONE or 
  $\sparsity(\tilde{h}) \leq 2\cdot\sparsity(\hat{h})$}{\KwRet{$\tilde{h}$}}\;
  $m \gets \lcm\{i_1,i_2,\ldots,i_k\}$\;
  \KwRet{$\hat{h}\cdot (x^m-1)$} 
\end{algorithm}

\begin{theorem}\label{thm:ratgenalg}
  Let $f\in\QQ[x]$ be a degree-$d$ non-original polynomial with no repeated
  cyclotomic factors. Given $f$ and integers $c$ and $t$,
  Algorithm~\ref{alg:sqfree} correctly computes a $t$-sparse multiple $h$
  of $f$ satisfying $\H(h)\leq c$, if one exists. 
  The sparsity of $h$ will be minimal over all multiples with
  height at most $c$.
  The algorithm requires
  $(d \log c)^{O(t)} \cdot 2^{O(t\log t)} \cdot (\log \H(f))^{O(1)}$
  bit operations.
\end{theorem}
\begin{proof}
  Step~\ref{alg:sqfree:factor} can be accomplished in the stated complexity
  bound using \cite{LLL82}.
  The cost of the remaining steps follows from basic arithmetic and
  Theorem~\ref{thm:a1}.
  Define $h$ to be sparsest multiple of $f$ of least degree that
  satisfies $\H(h)\leq c$. We have two cases:

  \begin{description}
    \item[Case 1: $\deg h \leq n$.]
      Then the computed $\tilde{h}$ must equal $h$. Furthermore, since this is
      the sparsest multiple, either $\hat{h}$ does not exist or 
      the sparsity of $\hat{h}$ is greater than or equal to the sparsity of
      $\tilde{h}$. So $h=\tilde{h}$ is correctly returned by the algorithm 
      in this case.
    \item[Case 2: $\deg h > n$.]
      Then, using Lemma~\ref{lem:gap}, since $f_D\mid h$, $h$ can be written
      $h = h_1 + x^i h_2$, for some $i>\deg h_1$, and $f_D$ divides both
      $h_1$ and $h_2$. By Theorem~\ref{thm:a1}, $\sparsity(\hat{h})$ must
      then be less than or equal to each of $\sparsity(h_1)$ and 
      $\sparsity(h_2)$. But since 
      $\sparsity(h) = \sparsity(h_1)+\sparsity(h_2)$, this means that
      the sparsity of $\hat{h}\cdot (x^m-1)$ is less than or equal to
      the sparsity of $h$, and hence this is a sparsest multiple.
  \end{description}
  \QED
\end{proof}

\subsection{An example}\label{sec:ratgenex}

Say we want to find a sparsest multiple, with coefficients
at most $1000$ in absolute
value, of the following
polynomial over $\ZZ[x]$.
\[
f = {x}^{10}-5{x}^{9}+10{x}^{8}-8{x}^{7}+7{x}^{6}-4{x}^{5}
+4{x}^{4}+{x}^{3}+{x}^{2}-2x+4
\]

Note that finding the \emph{sparsest} multiple would correspond to
setting $t=10$ in the algorithm (since the least-degree 11-sparse multiple
is $f$ itself).
To accomplish this, we
first factor $f$ using \citep{LLL82} and identify cyclotomic factors:
\[
f = 
\underbrace{(x^2-x+1)}_{\Phi_6} \cdot 
\underbrace{(x^4-x^3+x^2-x+1)}_{\Phi_{10}} \cdot 
\underbrace{({x}^{4}-3{x}^{3}+{x}^{2}+6x+4)}_{f_D}.
\]

Next, we calculate a degree bound from Theorem~\ref{thm:degbound}. 
Unfortunately, this bound is not very tight (despite being polynomial in the
output size); using $t=10$, $c=1000$, and $f$ given above,
the bound is $n \leq 11\,195\,728$.
So for this example, we will use
the smaller (but artificial) bound of $n\leq 20$.

The next step is to calculate the sparsest 5-sparse multiple of $f_D$ and 
10-sparse multiple of $f$ with
degrees at most 20 and heights at most 1000. 
Using Algorithm~\ref{alg:boundednc}, these are respectively
\begin{align*}
\hat{h} =&\ {x}^{12}+259{x}^{6}+64 \\
\tilde{h} =&\ {x}^{11}-3{x}^{10}+12{x}^{8}-9{x}^{7}+10{x}^{6}-4{x}^{5}+9{x}^{4}
+3{x}^{3}+8.
\end{align*}

Since the sparsity of $\hat{h}$ is less than half that of $\tilde{h}$,
a sparsest multiple is 
\begin{align*}
  h =&\ ({x}^{12}+259{x}^{6}+64)\cdot(x^{\lcm(6,10)}-1) \\
 =&\ {x}^{42}+259{x}^{36}+64{x}^{30}-{x}^{12}-259{x}^{6}-64.
\end{align*}

%%% Local Variables:
%%% mode: latex
%%% TeX-master: "alg_submit"
%%% End:

%-----------------------------------------------------------------------------%
\section{Sparse multiples over $\Fq$}
\label{sec:Fq-gen}
%-----------------------------------------------------------------------------%

We prove that for any constant $t$, finding the minimal
degree $t$-sparse multiple of an $f \in \Fq[x]$ is harder than
finding orders of elements in $\Fqe$. 
Order finding is reducible to integer factorization and to discrete logarithm, 
but reductions in the
other direction are not known for finite fields \citep{AdlMcC94}.
However, at least for prime fields and assuming the
Extended Riemann Hypothesis, 
a fast algorithm for order finding in finite fields
would give an efficient procedure for computing primitive elements
\citep{Wan59,Sho92}.  The latter problem
is regarded as ``one of the most important
unsolved and notoriously hard problems in the computational theory of
finite fields'' \citep{vzGShp99}.

Formal problem definitions are as follows:
\begin{description}
\item[\tSpMulFqn:] Given a polynomial $f\in\Fq[x]$ and an integer
  $n\in\NN$,
  determine if there
  exists a (nonzero) 2-sparse multiple $h\in\Fq[x]$ of $f$ 
  with $\deg h \leq n$.
\item[\OrderFqe:] Given an element $a \in \Fqe^*$ and an integer $n<q^e$,
  determine if there exists a positive integer $m\leq n$ such that $a^m=1$.
\end{description}
The problem \OrderFqe is well-studied (see for instance
\cite{Meijer96}), and has been used as a primitive in several
cryptographic schemes. Note that an algorithm to solve \OrderFqe will
allow us to determine the \emph{multiplicative order} of any $a\in\Fqe^*$
(the smallest nonzero $m$ such that $a^m=1$) with essentially the
same cost (up to a factor of $O(e\log q)$) by using binary search.

The reduction from \OrderFqe to \tSpMulFqn works as follows: Given an
instance of \OrderFqe, we first check if the order $o_a$ of $a$ is
less than $t$ by brute-force. Otherwise, we construct the minimal
polynomial $g_{a^i}$ (over $\Fq$) for each $a^0,a^1,a^2,\ldots,a^{t-1}$. We only
keep distinct $g_{a_i}$, and call the product of these distinct
polynomials $f_{a,t}$. We then run the \tSpMulFqn subroutine to search
for the existence of a degree $n$, $t$-sparse multiple of the
polynomial $f_{a,t}$.

\begin{theorem} \label{theorem:hardnessoftsparse} Let $a \in \Fq$ be
  an element of order at least $t$. Then the least degree $t$-sparse
  multiple of $f_{a,t}$ is $x^{o_a}-1$ where $o_a$ is the order of
  $a$. \end{theorem}
\begin{proof} It is easy to see that $x^{o_a}-1$ is a multiple of the given polynomial. We need to prove that it is actually the least-degree $t$-sparse multiple.

By equation~\eqref{eq:rootPowers} in Section 2, a degree $n$ multiple $h$ of $f_{a,t}$ corresponds to the following set of linear equations:
\begin{equation*}
\underbrace{
\left[\begin{array}{ccccc}
    1 & 1 & 1 & \cdots & 1\\
    1 & a & a^2 &\cdots & a^{n-1}\\
    1 & a^2 &a^4 &\cdots & a^{2n-2}\\
    \vdots & \vdots & \vdots & \vdots & \vdots\\
    1 & a^t& a^{2t} & \cdots & a^{tn-t}\\
\end{array}\right]}_{A(f_{a,t},n)}
\left[\begin{array}{c}
    h_0 \\ h_1 \\ [5mm] \vdots \\ h_{n-1}
 \end{array}\right]
= 0.
\end{equation*}

To prove that no $t$-sparse multiple $h$ of degree less than $o_a$ exists, it suffices to show that any $t$ columns of $A(f_{a,t},o_a-1)$ are linearly independent. Consider the $(t\times t)$-matrix corresponding to some choice of $t$ columns:
\begin{equation*}
B=\left[\begin{array}{cccc}
    1 & 1 & \cdots & 1\\
    a^{i_1} & a^{i_2} &\cdots & a^{i_t}\\
    \vdots & \vdots & \vdots & \vdots\\
    a^{ti_1} & a^{ti_2} &\cdots & a^{ti_t}\\
\end{array}\right].
\end{equation*}

This Vandermonde matrix $\mathbf{B}$ has determinant $\prod_{1\leq j < k \leq t} (a^{i_k}-a^{i_j})$ which is nonzero since $i_j < i_k < o_a$ and hence $a^{i_j} \neq a^{i_k}$. Thus the least-degree $t$-sparse multiple of the given polynomial is $x^{o_a}-1$.
\QED
\end{proof}

Of cryptographic interest is the fact that the order-finding polynomials
in the reduction above are sufficiently dense in $\Fq[x]$ that the
reduction also holds in the average case. That is, an algorithm for
sparsest multiples that is polynomial-time on average would imply an
average case polynomial-time algorithm for order finding in $\Fqd$.

Next we give a probabilistic algorithm for finding the
least degree binomial multiple for polynomials $f \in \Fq$. This
algorithm makes repeated calls to an \OrderFqe (defined in the
previous section) subroutine. Combined with the hardness result of the
previous section (with $t$=2), this characterizes the
complexity of finding least-degree binomial multiples in terms of the
complexity of \OrderFqe, upto randomization.

Algorithm~\ref{alg:binomialFiniteFields} solves the binomial multiple
problem in $\Fq$ by making calls to an \OrderFqe procedure that
computes the order of elements in extension fields of $\Fq$. Thus
\BinSpMulFq reduces to \OrderFqe in probabilistic polynomial time.
Construction of an irreducible polynomial (required for finite field
arithmetic) as well as the factoring step in the algorithm make it
probabilistic.

\begin{algorithm}[htbp]
\caption{Least degree binomial multiple of $f$ over $\Fq$}
\label{alg:binomialFiniteFields}
\KwIn{$f\in\Fq[x]$}
\KwOut{The least degree binomial multiple $h$ of $f$}

Factor $f=x^bf_1^{e_1} \cdot f_2^{e_2} \cdot f_\ell^{e_{\ell}}$ for irreducible $f_1,\ldots,f_\ell\in\Fq[x]$, and set $d_i \gets \deg f_i$
\label{bff:factor}

\For{$i = 1,2,\ldots,\ell$} {
  $a_i \gets x \in\Fq[x]/(f_i)$, a root of $f_i$ in the extension
  $\mathbb{F}_{q^{d_i}}$\;
  Calculate $o_i$, the order of $a_i$ in  $\Fq[x]/(f_i)$.
}

$\displaystyle n_1 \gets \lcm(\{o_i/\gcd(o_i,q-1)\})$ for all $i$ such that $d_i > 1$\label{bff:n1}\;
$n_2 \gets \lcm(\{order(a_i/a_j)\})$ over all $1 \leq i,j \leq u$\label{bff:n2}\;
$n \gets \lcm(n_1,n_2)$\label{bff:n}\;
$\tilde{h} \gets (x^n-a_1^n)$\;
$e \gets \lceil \log_p{\max{e_i}}\rceil$, the smallest $e$ such that $p^e \geq e_i$ for all $i$\;
\KwRet{$h=x^b(x^n-a_1^n)^{p^e}$}

\end{algorithm}

\begin{theorem}\label{thm:binfqalg}
  Given $f\in\Fq[x]$ of degree $d$,
  Algorithm~\ref{alg:binomialFiniteFields} correctly computes a
  binomial multiple $h$ of $f$ with least degree.  It uses at most
  $d^2$ calls to a routine for order finding in $\Fqe$, for various
  $e\leq d$, and $d^{O(1)}$ other operations in $\Fq$.  It is
  probabilistic of the Las Vegas type.
\end{theorem}
\begin{proof}
As a first step, the algorithm factors the given polynomial into
irreducible factors. Efficient probabilistic algorithms for factoring
polynomials over finite fields are well-known (\cite{MCA2003}).

First, suppose the input polynomial $f$ is irreducible, i.e. $\ell=e_1=1$
in Step~\ref{bff:factor}.
Then it has the form
$f=(x-a)(x-a^q)\cdots(x-a^{q^{d-1}})$ for some $a \in \Fqd$,
where $d=\deg f$. If $f=(x-a)$, the least-degree binomial multiple is $f$ itself. Therefore, assume that $d>1$. Let the least-degree binomial multiple (in $\Fq[x]$) be $x^n-\beta$

Since both $a$ and $a^q$ are roots of $(x^n-\beta)$, we have that $a^n=a^{nq}$ and $a^{n(q-1)}=1$. Thus, the order $o_a$ of $a$ divides $n(q-1)$. The minimal $n$ for which $o_a \mid n(q-1)$ is $n=\frac{o_a}{\gcd(o_a,q-1)}$. Since this $n$ ensures that $a^n=a^{nq}$, it also simultaneously ensures that each $a^{q^i}$ is also a root. 

Notice that this $n$ equals $n_1$ computed on Step~\ref{bff:n1}, and
$n_2$ computed on Step~\ref{bff:n2} will equal 1, so the algorithm is correct in
this case.

Now suppose the input polynomial $f$ is reducible.
The factorization step factors $f$ into irreducible factors $f=f_1^{e_1}f_2^{e_2}\cdots f_{\ell}^{e_{\ell}}$. Let $\check{f}=f_1f_2\cdots f_{\ell}$ denote the squarefree part of $f$.

Being irreducible, each $f_i$ has the form $f_i(x)=
(x-a_i)(x-a_i^q)\cdots(x-a_i^{q^{d_i-1}})$ for some $a_i \in \Fqd$,
and $d_i=\deg f_i$. We make two observations:

\begin{itemize}
\item[$\bullet$] If $\check{f}(x)\mid x^n-a$ for some $a \in \Fq$, we have that $a_i^n=a_j^n$ for all $1 \leq i,j \leq \ell$, and hence that $(\frac{a_i}{a_j})^n=1$. Thus $order(\frac{a_i}{a_j}) \mid  n$. The least integer satisfying these constraints
is $n_2$ computed on Step~\ref{bff:n2}.

\item[$\bullet$] As before for the case when the input polynomial is irreducible and of degree more than one: $d_i > 1$ implies that $\frac{o_i}{\gcd(o_i,q-1)} \mid n$ for $o_i$ the order of $a_i$.
The least integer satisfying these constraints is $n_1$ computed on Step~\ref{bff:n1}.
\end{itemize}

The minimal $n$ is the least common multiple of all the divisors obtained from the above two types of constraints, which is exactly the value computed on
Step~\ref{bff:n}. 
The minimal degree binomial multiple of $\check{f}$ is $x^n-a^n_1$. 

It is easily seen that for the smallest $e$ such that $p^e \geq e_i$,
$(x^n-a^n)^{p^e}$ is a binomial multiple of $f$.  We now show that it
is actually the minimal degree binomial multiple.  Specifically, let
$e$ be the smallest non-negative integer such that $p^e \geq
\max{e_i}$; we show that the minimal degree binomial multiple of $f$
is $(x^n-a^n_i)^{p^e}$ for $n$ obtained as above.  

Let the minimal degree binomial multiple of $f$ be
$x^{\hat{n}}-b$. Factor $\hat{n}$ as $\hat{n}=\check{n}p^c$ for
maximal $c$, and write $(x^{\hat{n}}-b)$ as
$(x^{\check{n}}-b^{1/p^c})^{p^c}$. The squarefree part of $f$,
$\check{f}$ divides $(x^{\check{n}}-b^{1/p^c})$, and hence (by
constraints on and minimality of $n$) $(x^n-a_1^n)\mid
(x^{\check{n}}-b^{1/p^c})$. Thus $\check{n} \geq n$.

Since $c$ is chosen maximally, $p$ does not divide $\check{n}$, and
hence $x^{\check{n}}-b^{1/p^c}$ is squarefree. Using this and the fact
that $f$ divides $(x^{\check{n}}-b^{1/p^c})^{p^c}$, it is seen that
$p^c \geq e_i$ holds for all $e_i$, and hence $p^c \geq p^e$. This,
along with $\check{n} \geq n$, completes the proof that
$(x^n-a^n_i)^{p^e}$ is the minimal degree binomial multiple of $f$,
which completes the proof of the theorem.
\QED
\end{proof}

%%% Local Variables: 
%%% mode: latex
%%% TeX-master: "alg_submit"
%%% End: 

%-----------------------------------------------------------------------------%
\section{Conclusion and Open Problems} \label{sec:conclusion}
%-----------------------------------------------------------------------------%

To summarize, we have presented an efficient algorithm to
compute the least-degree binomial multiple of any rational polynomial.
We can also compute $t$-sparse multiples of rational polynomials that do not
have repeated cyclotomic factors, for any fixed $t$, and given a bound on
the height of the multiple.

We have also shown that, even for fixed $t$,
finding a $t$-sparse multiple of a degree-$d$
polynomial over
$\Fq[x]$ is at least as hard as finding the orders of elements in $\Fqd$.
In the $t=2$ case, there is also a probabilistic reduction in the other direction, so
that computing binomial multiples of degree-$d$ polynomials over $\Fq[x]$ probabilisticly reduces to order
finding in $\Fqd$.

Several important questions remain unanswered. 
Although we have an unconditional algorithm to compute binomial multiples
of rational polynomials, computing $t$-sparse multiples for fixed $t\geq 3$
requires an a priori height bound on the output as well as the requirement
that the input contains no repeated cyclotomic factors. Removing these
restrictions is desirable (though not necessarily possible).

Regarding lower bounds, we know that computing $t$-sparse 
multiples over finite fields is at least
as hard as order finding, a result which is tight (up to randomization) for $t=2$, but for larger
$t$ we believe the problem is even harder. 
Specifically, we suspect that computing $t$-sparse
multiples is \NP-complete over both $\QQ$ and $\Fq$, 
when $t$ is a parameter in the input.

%%% Local Variables:
%%% mode: latex
%%% TeX-master: "spmul"
%%% End:

%-----------------------------------------------------------------------------%
\subsection*{Acknowledgments}
The authors would like to thank John May, Arne Storjohann, and the anonymous referees
from ISAAC 2010 for their careful reading and useful observations on
earlier versions of this work.

\def\bibsep{4pt}

\bibliography{sparsemult}

\appendix
\section{Finding short $l_\infty$ vectors in lattices}
\label{app:aks}

In Section~\ref{sec:aks}, we presented Algorithm~\ref{alg:aksMod} 
to find the shortest
$l_\infty$ vector in the image of an integer matrix. 
This appendix
is devoted to proving the correctness of this algorithm,
culminating in the proof of Theorem~\ref{theorem:aksWorks!}.
Again, the results here are due to the presentation of
\cite{AjtaiKS2001} by \cite{Reg04}, with modifications to
accommodate the infinity norm.

For any lattice $\L$, define $s(\L)=\min_{v\in\L}\norm{v}_2$ to be the
least $l_2$ norm of any vector in $\L$.  If $\L$ satisfies $2\leq
s(\L) < 3$, and $B$ is a basis for $\L$,
then we will show that the main for loop in
Steps~\ref{step:aksforstart}--\ref{step:aksforend} of 
Algorithm~\ref{alg:aksMod} finds a vector in $\L$
with minimal $l_\infty$ norm, with high probability.  
The for loop on line~\ref{step:aksforstart} adapts this to work for any
lattice by scaling.
More precisely, given
a lattice $\L$, we first run the algorithm of \cite{LLL82} to get an
approximation $\lambda$ for the shortest $l_2$ vector in $\L$
satisfying $s(\L) \leq \norm{\lambda}_2\leq 2^n s(\L)$. 
For each $k$ from $1$ to $2n$, we then run the for loop with basis
$B_k$ for the lattice $(1.5^k/\norm{\lambda}_2)\cdot \L$.
For some $k$ in this range, $2 \leq s(B_k) < 3$ must hold, and we will
show that for this $k$, the vector $v_k$ set on
Step~\ref{step:aksforend} is the $l_\infty$ shortest vector in the image
of $B_k$ with high probability. 
For every $k$, $v_k$ is a vector in the
image of $B_k$, and hence it suffices to output the shortest $l_\infty$ vector
among $\{(\norm{\lambda}_2/1.5^k)v_k\}$ on Step~\ref{step:aksret}.

We will now prove that the vector $v_k$ set on Step~\ref{step:aksforend}
is with high probability the shortest $l_\infty$ vector in the image of
$B$, when $B$ is a basis for
a lattice $\L$ such that $2 \leq s(\L) < 3$.

To find the shortest $l_\infty$ vector in a lattice, it suffices to
consider all lattice vectors of $l_2$ norm at most $\sqrt{n}$ times
the norm of the shortest $l_2$ vector. Algorithm~\ref{alg:aksMod}
achieves this by running the main body of the loop with different
values of $\gamma$. In a particular iteration of the outermost loop,
with high probability, the algorithm encounters all lattice vectors
$v$ with $l_2$ norm satisfying $(2/3)\cdot\norm{v}_2 \leq \gamma <
\norm{v}_2$. Call all such $v$ \emph{interesting}. By iterating over a
suitable range of $\gamma$, it returns the shortest $l_\infty$ vector
among all the \emph{interesting} vectors, which with high probability
is the shortest $l_\infty$ vector in the lattice.

For a particular iteration of the loop (with a fixed $\gamma$), the
algorithm uniformly samples a large number of vectors from an
appropriately sized ball. In fact, the algorithm works even if an
almost-uniform sampling over rational vectors with bit lengths bounded
by $(\log \norm{B} + n)^{O(1)}$ is performed. This is because the size
of sufficiently small lattice vectors is only a polynomial in the size
of the basis vectors. For the rest of this subsection, ``arithmetic
operations'' means operations with rational numbers of this size.

After sampling, the algorithm performs a series of \emph{sieving}
steps to ensure that at the end of these steps the algorithm is left
with lattice vectors of sufficiently small $l_2$ norm. Using a
probabilistic argument, it is argued that all \emph{interesting}
vectors are obtained.

The following lemma proves the correctness of the sieving steps. These
correspond to Steps \ref{step:sievingStart} to \ref{step:sievingEnd}
of the algorithm. At the end of this sieving, the algorithm produces a
set $J$ of size at most $5^n$.

\begin{lemma}
  \label{lemma:sieveLemma} 
  Given $S\subseteq\{1,\ldots,m\}$ such that for all $i\in S$,
  $y_i\in\RR^n$ and $\norm{y_i}_2 \leq r$, Steps
  \ref{step:sievingStart}--\ref{step:sievingEnd} efficiently compute
  the following: a subset $J \subseteq S$ of size at most $5^n$ and a
  mapping $\eta: S\setminus J \rightarrow J$ such that
  $\norm{y_i-y_{\eta_i}}_2 \leq r/2$.
\end{lemma}
\begin{proof} 
  Initially the set $J$ is empty. The algorithm iterates over the
  points $y_i$ with $i\in S$, adding $i$ to $J$ only if $min_{j \in
    J}(\norm{y_j-y_i}_2)>r/2$. For $i \notin J$, it sets $\eta_i$ to a
  $j\in J$ such that $\norm{y_j-y_i}_2\leq r/2$.

  It is clear that this procedure runs in polynomial time. To see that
  the size of $J$ is at most $5^n$, note that all the balls of radius
  $R/4$ and centered at $y_j$ for $j \in J$ are disjoint by
  construction of $J$. Also, these balls are contained in a ball of
  radius $R+R/4$ since $\norm{y_i}_2 \leq R$. Thus the total number of
  disjoint balls, and hence the size of $J$, can be bounded above by
  comparing the volumes: $|J| \leq (  (5R/4)/(R/4))^n = 5^n$.
  \QED
\end{proof}

The algorithm views every sampled vector $x_i$ as a perturbation of a
lattice vector $x_i-y_i$ for some $y_i$. The idea is the following:
initially $y_i$ is calculated so that $x_i$ is a perturbation of some
large lattice vector. Iteratively, the algorithm either obtains
shorter and shorter lattice vectors corresponding to $x_i$, or
discards $x_i$ in some sieving step. At all stages of the algorithm,
$x_i-y_i$ is a lattice vector. The following two lemmas concretize
these observations.

\begin{lemma} 
  \label{lemma:xMinusyIsInLattice} 
  $\{y_i\}$ can be found efficiently in
  Step~\ref{step:initialEstimateForY}; and $\{x_i-y_i\} \subseteq
  \mathcal{L}$.
\end{lemma}

\begin{proof}
  For a fixed $x_i$, $y_i$ is set to $(x_i \bmod \mathcal{P}(B))$
  where $\mathcal{P}(B)$ denotes all vectors contained in the
  parallelogram $\big\{\sum_{i=1}^n \alpha_i b_i \ |\ 0 \leq \alpha_i
  < 1\big\}$, with $b_i$ being the given basis vectors. Thus $y_i$ is
  the unique element in $\mathcal{P}(B)$ such that $y_i=x_i-v$ for $v
  \in \mathcal{L}$. From this definition of $y_i$, we get that
  $x_i-y_i \in \mathcal{L}$ for every $i$.

  To calculate $y_i$ efficiently, simply represent $x_i$ as a rational
  linear combination of the basis vectors $\{b_i\}$ and then truncate
  each coefficient modulo $1$.
  \QED
\end{proof}

\begin{lemma} 
  $Y_\gamma \subseteq \mathcal{L} \cap \ball{0}{3\gamma+1}$.
\end{lemma}
\begin{proof}
  By Lemma~\ref{lemma:xMinusyIsInLattice}, $(x_i - y_i) \in
  \mathcal{L}$ for all $i\in S$ before the start of the loop. It needs
  to be proved that the same holds after the loop, and furthermore,
  all the resulting lattice vectors lie in $\ball{0}{3\gamma+1}$.
  Whenever the algorithm modifies any $y_i$, it sets it to
  $y_i+x_{\eta(i)}-y_{\eta(i)}$; and thus a lattice vector $(x_i-y_i)$
  changes into $(x_i-y_i)-(x_{\eta(i)}-y_{\eta(i)})$. Since both of
  the terms are lattice vectors, so is their difference.  Thus
  $Y_\gamma \subseteq \mathcal{L}$.
 
  We will now show that the invariant $\norm{y_i}_2 \leq r$ is
  maintained at the end of every iteration. This suffices to prove
  that $x_i-y_i \in \ball{0}{3\gamma+1}$ because $x_i \in
  \ball{0}{\gamma}$ and $\norm{y_i}_2 \leq 2\gamma+1$ by the loop
  termination condition.

  Initially, $y_i= \sum_{j=1}^n \alpha_j b_j$ for some coefficients
  $\alpha_j$ satisfying $0 \leq \alpha_j < 1$. Thus $\norm{y}_2 \leq
  \sum_j \norm{b_j}_2 \leq n \max_j \norm{b_j}_2$, the initial value
  of $r$. Consider now the result of the change $y_i \rightarrow
  y_i+x_{\eta_i}-y_{\eta_i}$. We have that
  $\norm{y_i+x_{\eta_i}-y_{\eta_i}}_2 \leq
  \norm{y_i-y_{\eta_i}}_2+\norm{x_{\eta_i}}_2$. The first of these
  terms is bounded by $r/2$ because of choice of $\eta_i$ in
  Lemma~\ref{lemma:sieveLemma}. From $\norm{x_i}_2 \leq \gamma$, we
  get that $\norm{y_i}_2 \leq r/2+\gamma$. Since the value of $r$ gets
  updated appropriately, the invariant $\norm{y_i}_2 \leq r$ is
  maintained at the end of the loop.
  \QED
\end{proof}

The following crucial lemma says that $Y_\gamma$ can be used to
compute all \emph{interesting} vectors:

\begin{lemma} 
  \label{lemma:mainLemma} 
  Let $v \in \mathcal{L}$ be a lattice vector such that
  $(2/3)\cdot\norm{v}_2 \leq \gamma < \norm{v}_2$.  Then, with
  probability at least $1-1/2^{O(n)}$, $\exists w\in\mathcal{L}$ such
  that $Y_\gamma$ contains both $w$ and $w\pm v$.
\end{lemma}

Using this lemma, we can prove our main theorem, which we restate from
Section~\ref{sec:aks}:
\begin{theorem*} {\upshape (Theorem~\ref{theorem:aksWorks!})}\\
  \aksworks
\end{theorem*}
\begin{proof}
  Define $B_k$ to be the basis $B$ set on Step~\ref{step:akssetB} at
  iteration $k$ through the for loop on line~\ref{step:aksforstart}.
  For correctness, consider the iteration $k$ such that the lattice $\L$
  of $B_k$ satisfies $2\leq s(\L) < 3$, which we know must exist from
  the discussion above.
  
  Denote by $v_\infty$ the shortest nonzero vector in $\L$
  under the $l_\infty$ norm. We have that $l_2(v_\infty) \leq
  \sqrt{n}\cdot l_\infty(v_\infty) \leq \sqrt{n}\cdot l_\infty(v) \leq
  \sqrt{n}\cdot l_2(v)$, for any nonzero vector $v\in\L$.  Hence, the
  $l_2$ norm of the shortest $l_\infty$ vector is at most $\sqrt{n}$
  times the $l_2$ norm of the shortest $l_2$ vector.

  Since the length $s(\mathcal{L})$ of the shortest $l_2$ vector is
  assumed to satisfy $2 \leq s(\L) < 3$, we have that the $l_2$ norm of
  $v_\infty$ satisfies $\norm{v_\infty}_2 < 3 \sqrt{n}$.  Therefore at
  least one iteration of the while loop on line~\ref{step:akswhile} has
  $(2/3)\cdot \norm{v_\infty}_2 \leq \gamma < \norm{v_\infty}_2$, and
  by Lemma~\ref{lemma:mainLemma}, with high probability some
  $Y_\gamma$ contains $w$ and $w\pm v_\infty$ for some $w\in\L$. Since
  the algorithm computes the differences of the vectors in $Y_\gamma$,
  it sets $v_k$ to $v_\infty$ on Step~\ref{step:aksforend}
  with high probability.

  For the cost analysis, consider a single iteration of the 
  while loop on line~\ref{step:akswhile}.
  The value of $r_0$ is bounded by $(n\cdot
  \norm{U})^{O(1)}$. The value of $m$ is bounded by $2^{O(n\log
    \gamma)}\log r_0$, which is in turn bounded by $2^{O(n\log
    n)}\cdot \norm{U}^{O(1)}$ because $\gamma \in O(\sqrt{n})$.  Since
  the number of sieving steps is $O(\log r_0) \in O(m)$, the total
  cost of a single iteration of the while loop
  is $m^{O(1)}$. The total number of
  iterations of the while loop is $O(\log n) \in O(m)$, and there are
  exactly $2n\in O(m)$ iterations of the outer for loop.
  Each arithmetic operation costs $(n\cdot \norm{U})^{O(1)}\in O(m)$, so
  the total cost is $m^{O(1)}$, which gives the stated bound.
  \QED
\end{proof}

To prove Lemma~\ref{lemma:mainLemma}, a probabilistic argument will be
employed. The proof can be broken into three steps. First, we identify a
set of \emph{good} points from the sampled points, and argue that this
set is large. Next, we argue that there must exist a lattice point which
corresponds to numerous \emph{good} points. Finally, we argue that an
imaginary probabilistic step does not essentially change the behaviour
of the algorithm. Combined with the existence of a lattice point
corresponding to many \emph{good} points, this imaginary step allows
us to argue that the algorithm encounters both $w$ and $w\pm v$ for an
appropriate \emph{interesting} $v$.

Let $v$ be an \emph{interesting} lattice vector. That is,
$(2/3)\cdot d \leq \gamma < d$ for $d=\norm{v}_2$. For the iteration
where the algorithm uses a value of $\gamma$ in this range, we will
denote by $C_1$ the points in the set $\ball{v}{\gamma}\cap
\ball{0}{\gamma}$.  Similarly, $C_2=\ball{-v}{\gamma} \cap
\ball{0}{\gamma}$. By choice of $\gamma$, $C_1$ and $C_2$ are
disjoint. We will call the points in $C_1 \cup C_2$ \emph{good}. The
following lemma shows that probability of sampling a \emph{good} point
is large.

\begin{lemma} 
  \label{lemma:probGoodPoint} 
  $\mathbf{Pr}[x_i \in C_1] \geq 2^{-2n}$.
\end{lemma}
\begin{proof} 
  The radius of both $\ball{0}{\gamma}$ and $\ball{v}{\gamma}$ is
  $\gamma$. The distance between the centers is $d=\norm{v}_2$. Thus
  the intersection contains a sphere of radius $\gamma-d/2$ whose
  volume gives a lower bound on the volume of $C_1$.  Comparing with
  the volume of $\ball{0}{\gamma}$ and using the fact that $\gamma
  \geq (2/3)\cdot d$, we get that
  \[
  \mathbf{Pr}[x_i \in C_1] \geq
  \frac{\Vol(\ball{0}{\gamma-d/2})}{\Vol(\ball{0}{\gamma})} \geq
  \left(\frac{\gamma/4}{\gamma}\right)^n=2^{-2n}.
  \]
  \QED
\end{proof}

Informally, the following lemma says that if $S$ is large at the end
of the inner loop, the set $\{x_i-y_i\}$ has many repetitions and
hence is never very large.

\begin{lemma} 
  \label{lemma:numLatticeRemaining} 
  $|Y_\gamma| \leq (3\gamma+2)^n$.
\end{lemma}
\begin{proof} 
  The points in $\mathcal{L}$ are separated by a distance of at least
  $2$ since we assumed $s(\mathcal{L}) \geq 2$. Hence balls of radius
  $1$ around each lattice point are pairwise disjoint. If we consider
  only the balls corresponding to points in $Y_\gamma$, all of them
  are contained in a ball of radius $3\gamma+2$ since $Y_\gamma
  \subseteq\ball{0}{3\gamma+1}$ by
  Lemma~\ref{lemma:xMinusyIsInLattice}. Thus the total number of
  points in $Y_\gamma$ is at most
  $\Vol(\ball{0}{3\gamma+2})/\Vol(\ball{0}{1})=(3\gamma+2)^n$.
\end{proof}

The following lemma argues that there must be a lattice point
corresponding to many \emph{good} points.

\begin{lemma} 
  \label{lemma:lotsOfGoodPoints}
  With high probability, there exists $w \in Y_\gamma$ and $I
  \subseteq S$ such that $|I| \geq 2^{3n}$, and for all $i \in I$,
  $x_i \in C_1 \cup C_2$ and $w = x_i-y_i$.
\end{lemma}
\begin{proof} 
  Since $\mathbf{Pr}[x_i \in C_1 \cup C_2]$ is at least $2^{-2n}$ by
  Lemma~\ref{lemma:probGoodPoint}, and the number of points sampled is
  $\ceil{2^{(7+\ceil{\log(\gamma)})n} \log r_0}$, the expected number
  of \emph{good} points sampled at the start is at least
  $2^{(5+\ceil{\log(\gamma)})n} \log r_0$. The loop performs $\log
  r_0$ iterations removing (by Lemma~\ref{lemma:sieveLemma}) at most
  $5^n$ points per iteration. The total number of \emph{good} points
  remaining in $S$ after the sieving steps is
  $(2^{(5+\ceil{\log(\gamma)})n}-5^n) \log r_0 \geq
  2^{(2+\ceil{\log(\gamma)})n}\log r_0$ since $5^n \leq 2^{3n}$.

  By Lemma~\ref{lemma:numLatticeRemaining}, $|Y_\gamma| \leq
  (3\gamma+2)^n$. Since $3\gamma+2 \leq 4\gamma$ for $\gamma \geq
  (3/2)^2$, $|Y_\gamma| \leq 2^{(2+\log(\gamma))n}$. Hence, there
  exists a $w \in Y_\gamma$ corresponding to at least
  $2^{(4+\ceil{\log(\gamma)})n}\log r_0/2^{(2+\log(\gamma))n}
  \geq 2^{3n}$ \emph{good} points.
  \QED
\end{proof}

The final step in the analysis is to argue that for such a $w \in
Y_\gamma$, we must also have that $w\pm v \in Y_\gamma$ with high
probability for an \emph{interesting} $v \in \mathcal{L}$.

\noindent\textit{Proof of Lemma~\ref{lemma:mainLemma}}

Consider the iteration where $\gamma$ satisfies $(2/3)\cdot\norm{v}
\leq \gamma < \norm{v}$ for an \emph{interesting} lattice vector $v$.

It can be easily seen that $x \in C_1$ if and only if $x-v \in
C_2$. Consider an imaginary process performed just after sampling all
the $x_i$. For each $x_i \in C_1$, with probability $1/2$, we replace
it with $x-v \in C_2$. Similarly, for each $x \in C_2$, we replace it
with $x+v \in C_1$. (This process cannot be performed realistically
without knowing $v$, and is just an analysis tool.) The definition of
$y_i$ is invariant under addition of lattice vectors $v \in
\mathcal{L}$ to $x_i$, and hence the $y_i$ remain the same after this
process.

Since the sampling was done from the uniform distribution and since
$(x \in C_1) \leftrightarrow (x-v \in C_2)$ is a bijection, this
process does not change the sampling distribution.

We may postpone the probabilistic transformation $x_i \leftrightarrow
(x_i-v)$ to the time when it actually makes a difference. That is,
just before the first time when $x_i$ is used by the algorithm. The
algorithm uses $x_i$ in two places. For $i \in J$ during the
\emph{sieving} step, we perfom this transformation immediately after
computation of $J$. Another place where $x_i$ is used is the
computation of $Y_\gamma$. We perform this transformation just before
this computation.

In the original algorithm (without the imaginary process), by
Lemma~\ref{lemma:lotsOfGoodPoints}, there exists a point $w \in
Y_\gamma$ corresponding to at least $2^{3n}$ \emph{good} points. Let
$\{x_i\}$ be this large set of \emph{good} points. With high
probability, there will be many $x_i$ which remain unchanged, and also
many $x_i$ which get transformed into $x_i\pm v$. Thus, $Y_\gamma$
contains both $w$ and $w \pm v$ with high probability.  \QED

%%% Local Variables:
%%% mode: latex
%%% TeX-master: "alg_submit"
%%% End:

\end{document}